\documentclass[sigconf]{acmart}
\AtBeginDocument{%
  \providecommand\BibTeX{{%
    \normalfont B\kern-0.5em{\scshape i\kern-0.25em b}\kern-0.8em\TeX}}}

\copyrightyear{2022}
\acmYear{2022}
\setcopyright{acmlicensed}\acmConference[KDD '22]{Proceedings of the 28th ACM SIGKDD Conference on Knowledge Discovery and Data Mining}{August 14--18, 2022}{Washington, DC, USA}
\acmBooktitle{Proceedings of the 28th ACM SIGKDD Conference on Knowledge Discovery and Data Mining (KDD '22), August 14--18, 2022, Washington, DC, USA}
\acmPrice{15.00}
\acmDOI{10.1145/3534678.3539238}
\acmISBN{978-1-4503-9385-0/22/08}

\usepackage{amsmath,amsfonts,amsthm,bm} 
\usepackage{natbib}
\usepackage{tikz}
\usepackage{pgfplots}
\pgfplotsset{compat=1.16}
\usetikzlibrary{patterns}
\usepackage{subcaption}
\usepackage[font=small]{caption}
\usepackage{multirow}
\usepackage{array}
\usepackage{balance}
\usepackage{hyperref}

\usepackage{enumitem}

\usepackage{amsfonts}

\def\-{\mbox{-}}

\def\la{\langle}

\def\ra{\rangle}

\def\J{\mathcal{J}}
\def\V{\mathcal{V}}
\def\E{\mathcal{E}}

\def\bh{\bm{h}}
\def\Ep{\mathbb{E}}

\def\X{\mathcal{X}}
\def\Y{\mathcal{Y}}
\def\D{\mathcal{D}}
\def\I{\mathcal{I}}

\def\Jo{\mathbb{J}}
\def\Ma{\mathbb{M}}

\def\*{\star}

\newcommand{\argmin}{arg\,min}

\usepackage{xcolor}

\definecolor{Mulberry}{rgb}{0.77,0.29,0.55}
\definecolor{CadmiumOrange}{rgb}{0.93,0.53, 0.18}
\definecolor{ForestGreen}{rgb}{0.13, 0.55, 0.13}
\definecolor{WildStrawberry}{rgb}{0.5, 0.7, 0.2}

\makeatletter
\newcommand{\axissize}{\@setfontsize{\axissize}{5pt}{5pt}}
\makeatother

\makeatletter
\newcommand{\Labelsize}{\@setfontsize{\Labelsize}{5pt}{5pt}}
\makeatother

\theoremstyle{definition}
\newtheorem{definition}{Definition}
\theoremstyle{plain}

\newtheorem{proposition}{Proposition}

\DeclareMathOperator{\R}{\mathbb{R}}
\newcommand{\norm}[1]{\left\lVert#1\right\rVert}

\usepackage{array}
\newcolumntype{P}[1]{>{\centering\arraybackslash}p{#1}}

\pgfplotsset{linestyle/.style={%
        font=\rmfamily\Labelsize,
        width=0.26\textwidth,
        height=3.3cm,
        mark size=1.1pt,
        ylabel near ticks,
        xlabel near ticks,
        label style = {font=\scriptsize},
        tick label style = {font=\scriptsize, yshift=0.5ex},
        ylabel shift = -6 pt, 
        xlabel shift = -5 pt,
        title style={yshift=-1.2ex,font=\footnotesize},
        y tick label style={/pgf/number format/.cd,fixed,fixed zerofill,precision=2,/tikz/.cd},
        legend image post style={scale=0.5},
        every axis plot/.append style={semithick},
        legend style={font=\scriptsize, mark size=4pt},
        legend columns=1, legend style={/tikz/column 2/.style={column sep=0.5pt}},
        mark options={scale=1.7}}}
        
\pgfplotsset{barstyle/.style={%
    ybar,
    width=0.245\textwidth,
    label style={font=\scriptsize},
    tick label style={font=\scriptsize,yshift=0.5ex},
    axis x line*=bottom,
    axis y line*=none,
    height=3.3cm,
    title style={yshift=-0.5ex,font=\small},
    ylabel shift = -5 pt, 
    xlabel shift = -2 pt,
    ylabel near ticks,
    xlabel near ticks,
    tickwidth         = 3pt,
    bar width=0.7mm,
    legend style={font=\scriptsize, mark size=4pt},
    y tick label style={/pgf/number format/.cd,fixed,fixed zerofill,precision=2,/tikz/.cd},
    legend columns=2, legend style={/tikz/column 2/.style={column sep=0.5pt}},
    label style = {font=\small}
}}

\definecolor{mycolor1}{RGB}{31, 95, 139}
\definecolor{mycolor2}{RGB}{122, 81, 149}
\definecolor{mycolor3}{RGB}{239, 86, 117}
\definecolor{mycolor4}{RGB}{255, 166, 0}

\begin{document}
\fancyhead{}

\title{Detecting Arbitrary Order Beneficial Feature Interactions for Recommender Systems}


\author{Yixin Su}
\email{yixins1@student.unimelb.edu.au}
\affiliation{%
  \institution{The University of Melbourne}
  \country{Australia}
}

\author{Yunxiang Zhao}
\email{zhaoyx1993@163.com}
\affiliation{%
  \institution{Bejing Institute of Biotechnology}
  \country{China}
}

\author{Sarah Erfani}
\authornote{Corresponding Author.}
\email{sarah.erfani@unimelb.edu.au}
\affiliation{%
  \institution{The University of Melbourne}
  \country{Australia}
}

\author{Junhao Gan}
\email{junhao.gan@unimelb.edu.au}
\affiliation{%
  \institution{The University of Melbourne}
  \country{Australia}
}

\author{Rui Zhang}
\email{rayteam@yeah.net}
\affiliation{%
  \institution{\href{https://www.ruizhang.info/}{www.ruizhang.info}
  \country{China}}
}
\renewcommand{\shortauthors}{Yixin Su et al.}

\begin{abstract}

Detecting beneficial feature interactions is essential in recommender systems, and existing approaches achieve this by examining all the possible feature interactions. However, the cost of examining all the possible higher-order feature interactions is prohibitive (exponentially growing with the order increasing). Hence existing approaches only detect limited order (e.g., combinations of up to four features) beneficial feature interactions, which may miss beneficial feature interactions with orders higher than the limitation.
In this paper, we propose a hypergraph neural network based model named HIRS. HIRS is the first work that directly generates beneficial feature interactions of arbitrary orders and makes recommendation predictions accordingly.
The number of generated feature interactions can be specified to be much smaller than the number of all the possible interactions and hence, our model admits a much lower running time.
To achieve an effective algorithm, we exploit three properties of beneficial feature interactions, and propose deep-infomax-based methods to guide the interaction generation.
Our experimental results show that HIRS outperforms state-of-the-art algorithms by up to 5\% in terms of recommendation accuracy.
\end{abstract}


\maketitle

\section{Introduction}
Features interactions (e.g., co-occurrence of user/item attributes) are essential in providing accurate recommendation predictions \cite{blondel2016polynomial}. 
Pairwise feature interactions, as the most basic feature interaction form, have been utilized in many studies \cite{he2017neural,li2019fi,su2021neural}.
Recently, researchers found that more sophisticated forms of feature interactions (i.e., high-order interactions) help achieve more accurate recommendation predictions \cite{lian2018xdeepfm,liu2020autofis}.
For example, if we observe that male teenagers like Nolan's sci-fi movies, the interactions (order-4) between user gender, user age, movie director, and movie genre provide useful information for recommendation that pairwise feature interactions cannot.

\begin{figure}[t]
\centerline{\includegraphics[width=0.99\columnwidth]{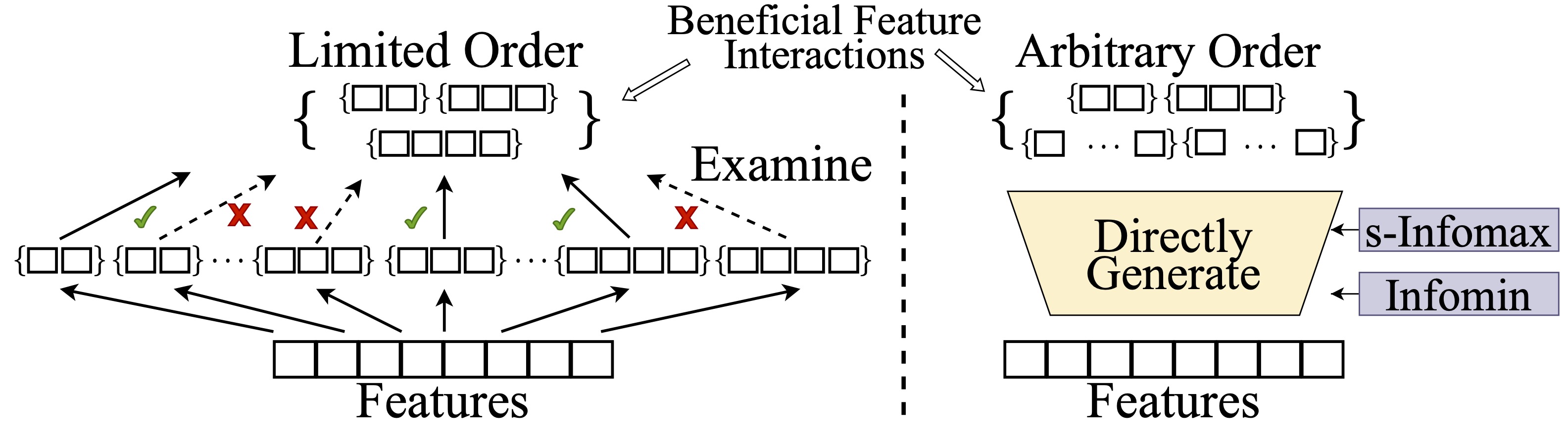}}
\caption{\small \textit{Left}: existing studies detect beneficial feature interactions within limited orders by examining all the possible feature interactions. \textit{Right}: we detect arbitrary order beneficial feature interactions by directly generating them with our deep-infomax-based methods (s-Infomax and Infomin).}
\label{fig:running_example}
\end{figure}

Since the number of all the possible feature interactions grows exponentially as the order increases, 
considering all the possible high-order feature interactions is prohibitive due to the high complexity. 
Meanwhile, considering the feature interactions irrelevant to the recommendation result may introduce noise and decrease the prediction accuracy \cite{louizos2017learning,wang2019doubly}.
Therefore, a practical solution to consider high-order feature interactions is to detect a small set of most \textit{beneficial} feature interactions and perform predictions based on the beneficial ones \cite{zeng2015novel}.
A set of feature interactions is considered most beneficial if learning from them results in a more accurate prediction than learning from other sets \cite{su2020detecting}.
However, existing interaction detection methods have to examine all the possible feature interactions to find the most beneficial feature interactions \cite{su2020detecting,liu2020autofis}.   
This leaves the complexity issue unsolved and they resort to detecting only limited order (e.g., combinations of up to four features) beneficial feature interactions (Figure \ref{fig:running_example} \textit{left}). 
Hence, it is still an urgent yet challenging problem to efficiently detect arbitrary order beneficial feature interactions.

In this work, we propose a method that can efficiently detect beneficial feature interactions of arbitrary orders. 
Our method learns a neural network that directly generates the beneficial feature interactions (Figure \ref{fig:running_example} \textit{right}).
The number of generated feature interactions can be specified to be much smaller than the number of all the possible interactions and hence, our method admits a much lower running time.
Then, we achieve recommendation predictions based on the generated feature interactions.
Specifically, we propose a hypergraph neural network based model called \underline{H}ypergraph \underline{I}nfomax \underline{R}ecommender \underline{S}ystem (HIRS).
HIRS learns a hypergraph representation for each data sample, 
with each feature in the data represented as a node and each beneficial feature interaction represented as a hyperedge. 
The beneficial feature interaction generation is conducted via hyperedge matrix prediction, and the recommendation prediction is conducted via hypergraph classification.

To achieve an effective beneficial feature interaction generation, we exploit three properties of beneficial feature interactions based on information theory:
\begin{itemize}[leftmargin = *]
\item \textit{Sufficiency}: all beneficial feature interactions together contain as much relevant information about the true output as possible. 
\item \textit{Low redundancy}: each beneficial feature interaction contains as less irrelevant information about the true output as possible.
\item \textit{Disentanglement}: any two beneficial feature interactions contain as less duplicated relevant information about the true output as possible.
\end{itemize}
Based on these properties, we leverage Deep Infomax \cite{hjelm2018learning} and propose supervised mutual information maximization (s-Infomax) and mutual information minimization (Infomin) methods, together with an $L_0$ activation regularization, to guide the interaction generation.

To effectively learn from the generated beneficial feature interactions for recommendation, we leverage the arbitrary-order correlation reasoning ability of hypergraph neural networks \cite{feng2019hypergraph}. We propose an interaction-based hypergraph neural network (IHGNN) to perform recommendation prediction based on the generated interactions via hypergraph classification.

We summarize our contributions as follows:
\begin{itemize}[leftmargin = *]
\item We propose a hypergraph neural network based model named HIRS.
\footnote{The implementation of our model can be found at 
\url{https://github.com/ruizhang-ai/HIRS_Hypergraph_Infomax_Recommender_System.git}.}.
HIRS is the first work that directly generates a set of beneficial feature interactions of arbitrary orders via hyperedge prediction and then performs recommendation prediction via hypergraph classification.
\item We define three properties of beneficial feature interactions to guide the proposed s-Infomax and Infomin methods for beneficial feature interaction generation. We further propose an interaction-based hypergraph neural network (IHGNN), which effectively learns from the generated interactions for recommendation.
\item Experimental results show that HIRS outperforms state-of-the-art feature interaction based models for recommendation. The further evaluations prove our model's ability in beneficial feature interaction generation.
\end{itemize}

\section{Related Work}
\label{sec:related_work}
\subsection{Feature Interaction based Recommendation}
Pairwise feature interactions have been widely used for recommendation, such as factorization machine (FM) based models \cite{rendle2010factorization,he2017neural,su2019mmf} and graph neural network based models \cite{su2020detecting,li2019fi}.
Higher-order feature interactions provide more information about the true output but have not been largely explored due to its high complexity. 
A multilayer perceptron (MLP) can implicitly model any order feature interactions for recommendation \cite{hornik1989multilayer,cheng2016wide,guo2017deepfm}, but implicit modeling is ineffective in learning feature interactions \cite{beutel2018latent}. 
\citet{lian2018xdeepfm} explicitly model all feature interactions under a limited order (e.g., up to four orders).
Besides modeling all feature interactions,
there are studies that detect a set of beneficial feature interactions for recommendation \cite{su2020detecting,liu2020autofis}.
However, they have to examine all the possible feature interactions, which leaves the complexity issue unsolved. 
Different from existing studies, we directly generate beneficial feature interactions without examining all the possible ones.

\subsection{Graph Neural Networks} 
Graph Neural Networks (GNNs) are powerful in correlation reasoning between entities \cite{xu2018powerful,zhao2021wgcn}. 
They have been applied in various research domains such as molecular prediction \cite{gilmer2017neural}, object relation learning \cite{chang2016compositional}, and user-item interaction modeling \cite{berg2017graph,wang2019neural}. 
Recently, GNNs show the power of feature interaction learning for recommendation \cite{li2019fi,su2020detecting,li2021graphfm}.
However, existing studies only capture pairwise feature interactions in a graph since each edge can only link to two nodes. 
We are the first to capture arbitrary order feature interactions and propose a Hypergraph Neural Network \cite{feng2019hypergraph} based model to learn arbitrary order feature interactions uniformly.

\subsection{Mutual Information Maximization}
Mutual information maximization is a critical technique in unsupervised representation learning \cite{linsker1988self}.
Deep Infomax (DIM, \citet{hjelm2018learning}) leverages a deep neural network-based method \cite{belghazi2018mutual} to maximize the mutual information between an image's local and global representations. A graph-based DIM is proposed to maximize the mutual information between node representations and graph representations \cite{velivckovic2018deep}.
For the first time, our work applies DIM on hypergraph that maximizes the mutual information between hyperedge representations and graph representations.
Contrastive learning is close to DIM but maximizes the mutual information between the two views of the same variable \cite{chen2020simple}.
\citet{tian2020makes} define properties of good views in contrastive learning. They hypothesize that each view contains all the information of an input variable about the output. 
We define the properties of beneficial feature interactions hypothesizing that each interaction only contains part of the information about the output.
\citet{khosla2020supervised} extend contrastive learning to a supervised setting. Our s-Infomax is similar to \cite{khosla2020supervised} but applies the supervised learning between local and global representations. 

\section{Preliminaries}

\subsection{Hypergraph Neural Networks}
Different from the edges in a typical graph, each edge in a hypergraph (i.e., hyperedge) can link to an arbitrary number of nodes. A hypergraph $G_H=\la \V, \E\ra$ contains a node set $\V=\{i\}_{i=1}^{m}$, where $m$ is the number of nodes. Each node $i$ is represented as a vector $\bm{v}_i\in\R^{d}$ of $d$ dimensions (e.g., node embedding).
The hyperedge set $\E=\{\bm{e}_j\}_{j=1}^{k}$ contains $k$ edges.
Each edge $\bm{e}_{j}$ is an $m$-dimensional binary vector ($e_{ji}\in\{0,1\}$), where $e_{ji}=1$ if the $j^{th}$ edge links to node $i$, and $e_{ji}=0$ otherwise.
We represent all the node vectors in a graph as a node representation matrix $\bm{V}\in\R^{m\times d}$ and the edge set as an incidence matrix $\bm{E}\in \{0,1\}^{m\times k}$.

In general, Hypergraph Neural Network (HGNN, \citet{feng2019hypergraph}) first generates a representation $\bm{h}_j$ for each hyperedge $j$ by aggregating the linked node vectors. Then, each node vector is updated by aggregating the edge representations linked to the node (called node patch representation). The procedure is:
\begin{equation}
\label{fun:HGNN}
\bm{h}_j = \phi(\{\bm{V}_i|\bm{E}_{ij}=1\}_{i\in\V}),\quad \bm{v}^{\prime}_{i}=\psi(\{\bm{h}_j|\bm{E}_{ij}=1\}_{\bm{e}_j\in\E}),
\end{equation}
where $\phi(\cdot)$ and $\psi(\cdot)$ are aggregation functions (e.g., element-wise mean), and $\bm{v}^{\prime}_{i}$ is the patch representation of node $i$.
The graph representation $\bm{c}\in\R^d $ can be generated by further aggregating the node patch representations:
\begin{equation}
\label{fun:HGNN_agg}
\bm{c} = \eta(\{\bm{v}^{\prime}_i\}_{i \in \V}),
\end{equation}
where $\eta(\cdot)$ is an aggregation function similar to $\phi(\cdot)$ and $\psi(\cdot)$.

\subsection{Deep Infomax}
Given two random variables $A$ and $B$, Mutual Information Neural Estimation (MINE, \citet{belghazi2018mutual}) estimates the mutual information $I(A;B)$ by training a discriminator that distinguishes the samples from their joint distribution $\Jo$ and from their marginal distribution $\Ma$. Specifically, MINE uses Donsker-Varadhan representation (DV, \citet{donsker1983asymptotic}) of KL-divergence as the lower bound to the mutual information:
\begin{equation}
\label{fun:DV_rep}
\small
\begin{split}
I(A;B) &= D_{KL}(\Jo||\Ma) \geqslant \hat{I}_{\omega}^{DV}(A;B)\\
&= \Ep_{\Jo}[T_{\omega}(a,b)]-\log\Ep_{\Ma}[e^{T_{\omega}(a,b)}],
\end{split}
\end{equation}
where $a$ and $b$ are the samples of $A$ and $B$, respectively, and $T_{\omega}$ is a classifier with parameters $\omega$.

In Deep Infomax \cite{hjelm2018learning} that aims to maximize a mutual information, the exact mutual information value is not important. Therefore, $\hat{I}_{\omega}^{DV}(A;B)$ can be conveniently maximized by optimizing a GAN-like objective function \cite{velivckovic2018deep,li2020graph}:
\begin{equation}
\label{fun:GAN_based_INFOMAX}
\small
\mathcal{L} = \Ep_{\Jo}[\log T_{\omega}(a,b)] + \Ep_{\Ma}[\log (1- T_{\omega}(a,b))].
\end{equation}

This objective function can be optimized by training the discriminator $T_{\omega}$ through a Binary Cross Entropy (BCE) loss \cite{velivckovic2018deep}.

\section{Problem Statement}
\label{sec:prel}

Denote by $\J$ the {\em universe} of features.
A {\em feature-value pair} is a name-value pair, denoted by $(o, w)$, where $o$ is the name of the feature and $w$ is the value of this feature.
For example, $(Male, 1)$ and $(Female, 1)$ mean the categorical features of Male and Female, respectively, where $Male \in \J$ and $Female \in \J$ are considered as two different features, and value 1 means the feature is in the data sample.
Let $\D:\X\times\Y$ be an input-output domain and $D=\{(\bm{x}_{n}, y_{n})\}_{i=1}^{N}$ is a set of $N$ input-output training pairs sampled from $\D$.
\begin{itemize}[leftmargin = *]
	\item $\bm{x}_n = \{(o, w)\}_{o\in J_n}$ is called a {\em data sample} consisting of a set of features, where $J_n \subseteq \J$.
	\item $y_n \in \{0,1\}$ is the {\em implicit feedback} (e.g., purchased, clicked) of the user on the item. 
\end{itemize}

\noindent
{\bf Feature Interaction based Recommendation.}
Given $(\bm{x}_{n}, y_{n})$, a recommendation model that \textit{explicitly} considers feature interactions first selects (either predefined or detected) $k_n$ feature interactions $\bm{I}_n=\{q_i\}_{i=1}^{k_n}$, where $q_i \subseteq \J_n$ indicates an interaction between the features in $q_i$. 
Then, a predictive model $F(\bm{x}_n,\bm{I}_n)=y^{\prime}_n$ is designed with a middle state $\bm{H}=\{\bm{h}_i\}_{i=1}^{k_n}$, where $\bm{h}_i\in\R^d$ is the high-level representation of $q_i$, and $y^{\prime}_n$ is the prediction of $y_n$.

\section{Our Proposed Method}
\label{sec:method}

In this section, we first give an overview of our proposed method. 
Then, we demonstrate HIRS in detail and the empirical risk minimization function of HIRS.
Finally, we provide analyses of our proposed model, including i) our proposed hypergraph neural network IHGNN can be easily generated to existing methods for feature interaction modeling; and ii) the interaction generation cost of HIRS.

\subsection{Model Overview}
\label{subsec:model_overview}
We start by introducing data representation with hypergraphs. Then we describe the overall structure of HIRS.

\subsubsection{Data Representation with Hypergraphs}
We focus on one input-output pair considering $k$ arbitrary order feature interactions. 
Existing GNN-based models can only represent pairwise feature interactions \cite{su2020detecting}. 
We propose a hypergraph form data representation that is more general in representing arbitrary order feature interactions. 
A hypergraph is an extension of a standard graph that each edge, i.e., hyperedge, can link to an arbitrary number of nodes.
We treat each data sample as a hypergraph, with each node being a feature that occurred in the data sample. 
Each hyperedge represents an interaction between all the linked features. 
Formally, for the feature set of each data sample input $\bm{x}=\{(o_i,w_i)\}_{i=1}^{m}$, it can be represented as a hypergraph $G_H\la \V,\E\ra$. The node set $\V=\bm{x}$.
The hyperedge set $\E$ consists of $k$ feature interactions (to be detected) that are most beneficial to the prediction accuracy (i.e., beneficial feature interactions). $\E$ can be represented as a hyperedge incidence matrix $\bm{E}\in\{0,1\}^{m\times k}$.

\subsubsection{The Overall Structure of HIRS}
HIRS contains two components: the first component performs recommendations predictions, and the second component performs s-Infomax and Infomin methods based on the three properties of beneficial feature interactions.
The recommendation prediction component of HIRS contains two modules:
one is a hyperedge prediction module that performs interaction generation; and the other is the IHGNN module that performs interaction learning. 
The recommendation prediction component takes a hypergraph without hyperedge as input, i.e., $G_H\la \V,\emptyset\ra$. 
A hyperedge prediction function $f_{gen}(\V)$ leverages the node information to generate $k$ beneficial feature interactions as the edge set $\E^{\prime}$. 
Then, IHGNN models the hypergraph with the generated edge set $G_H\la \V, \E^{\prime}\ra$ and outputs the recommendation prediction $y^{\prime}$.
The function combining the two modules is $f_{r}(\V;\bm{\rho})=f_{IHGNN}(G_H\la \V,f_{gen}(\V)\ra)$, where $\bm{\rho}$ are all the parameters in $f_{gen}$ and $f_{IHGNN}$.
While training, the other component containing s-Infomax and Infomin is conducted on the hyperedge representations ($\bm{h}$) and the graph representation ($\bm{c}$) in $f_{r}$, together with an $L_0$ activation regularization on the predicted hyperedges, to ensure the effectiveness of the interaction generation.
Figure \ref{fig:structure} demonstrates an overview of HIRS.
In the following subsections, we describe these two components in detail.

\begin{figure*}[t]
\centerline{\includegraphics[width=1.7\columnwidth]{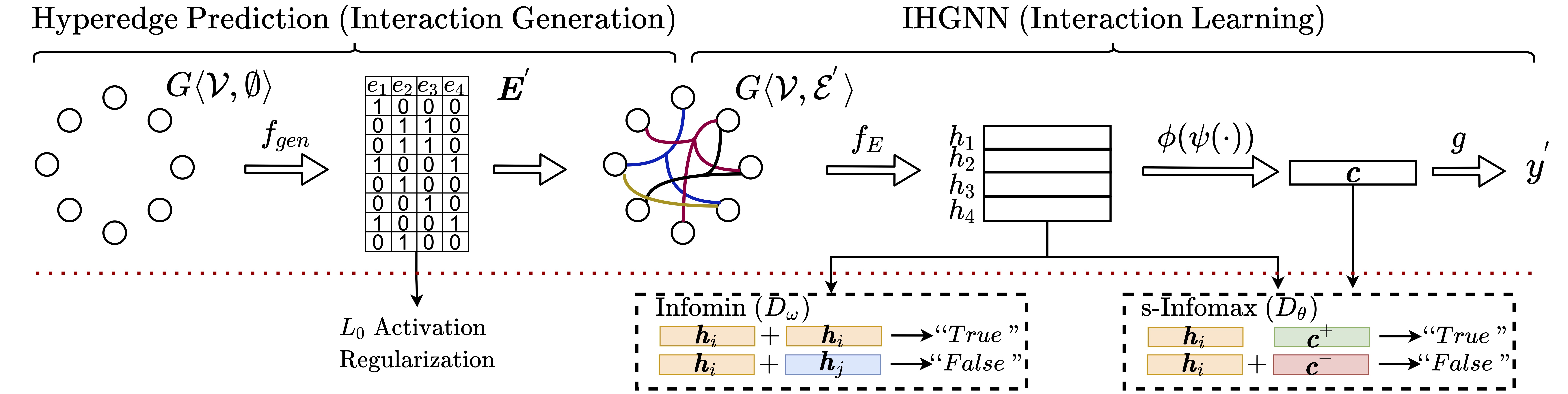}}
\caption{\small An Overview of HIRS. The part \textit{above} the red dotted line is the recommendation prediction component (RPC). The part \textit{below} the red dotted line is the s-Infomax and Infomin component and the $L_0$ activation regularization, which is only conducted during training.}
\label{fig:structure}
\end{figure*} 

\subsection{The Recommendation Prediction Component}

The recommendation prediction component (or RPC in short) of HIRS generates beneficial feature interactions via hyperedge prediction and then performs recommendation via IHGNN.

\subsubsection{Hyperedge Prediction}

The hyperedge prediction module $f_{gen}(\V)$ generates hyperedges $\E^{\prime}$ as beneficial feature interactions in an incidence matrix format $\bm{E}^{\prime}\in \{0,1\}^{m\times k}$. 
$\bm{E}^{\prime}_{ij}$ equals $1$ if the $i^{th}$ feature is in the $j^{th}$ beneficial feature interaction, and $0$ otherwise.

Specifically, each node is first mapped into a node vector of $d$ dimensions for hyperedge prediction: $\bm{V}^{e}_i=emb^{e}(o_i)\cdot w_i$, where $emb^{e}(\cdot)$ maps a feature into a $d$-dimensional embedding\footnote{$emb^{e}(\cdot)$ ensures a feature has the same embedding in different hypergraphs, and the embedding will be updated while training.}.
Then, for each node, we concatenate the vector containing the information of the node itself, and the vector summarizing the information of all other nodes in the hypergraph. The concatenated vector will be sent to an MLP to obtain a $k$-dimensional vector $\bm{u}_i$:
\begin{equation}
\label{fun:CrossRest}
\bm{u}_i = MLP([\bm{V}^{e}_i, \sum\nolimits_{j\neq i}\bm{V}^{e}_j]),
\end{equation}
where $[\cdot]$ is the concatenation operator.

Finally, we map $\bm{u}_i$ into binary values indicating the predicted hyperedge values $\bm{E}^{\prime}_i=\mathcal{B}(\bm{u}_i)$, where $\mathcal{B}$ is a Bernoulli distribution.
The predicted incidence matrix $\bm{E}^{\prime}\in \{0,1\}^{m\times k}$ is the concatenation of $\bm{E}^{\prime}_i$ of all the $m$ nodes.

However, directly representing $\bm{E}^{\prime}$ by binary values results in a non-differentiable problem while training with gradient descent methods \cite{shi2019pvae}.
Therefore, we replace the Bernoulli distribution with the hard concrete distribution \cite{maddison2016concrete}, which is differentiable and can approximate binary values.
Specifically, $\bm{u}_i$ is transferred into the hyperedge values of hard concrete distribution:
\begin{equation}
\label{fun:HardConcrete}
\small
\begin{split}
     & z \sim \mathcal{U}(0,1), \quad\quad  s_{ij} = Sig((\log z-\log (1-z)+\log(u_{ij}))/\tau), \\
     & \Bar{s}_{ij}  = s_{ij}(\delta-\gamma) + \gamma, \quad\quad \bm{E}^{\prime}_{ij}  = min(1, max(0, \Bar{s}_{ij})),
\end{split}
\end{equation}
where $z$ is a uniform distribution, $Sig$ is the Sigmoid function, $\tau$ is a temperature value and $(\gamma, \delta)$ is an interval with $\gamma<0$, $\delta>0$. 
We set $\gamma=-0.1$, $\delta=1.1, \tau=0.66$ following \citet{shi2019pvae}.

\subsubsection{IHGNN}
With the generated edge set $\E^{\prime}$ ($\bm{E}^{\prime}$ in incidence matrix form),
we propose an interaction-based hypergraph neural network (IHGNN) to learn $G_H\la\V,\E^{\prime}\ra$. 

First, we map each node in $\V$ into a node vector for IHGNN: $\bm{V}^{g}_i=emb^{g}(o_i)\cdot w_i$. Then, different from HGNN \cite{feng2019hypergraph} that models each hyperedge by linearly aggregating the node vectors, we use an MLP (a non-linear method) to model each hyperedge:
\begin{equation}
\label{fun:f_i}
\bm{h}_j = f_{E}(sum(\{\bm{V}^{g}_i|\bm{E}^{\prime}_{ij}=1\}_{i\in\V}))
\end{equation}
where $f_{E}$ is an MLP, $sum(\cdot)$ is an element-wise sum and $\bm{h}_j\in\R^{d}$ is a high-level representation of edge $j$.

The non-linearly modeling ensures IHGNN correctly and effectively models the generated interactions founded by statistical interactions \cite{su2020detecting}, which is theoretically justified in Appendix \ref{appx:prop2_proof}.
Then, each node representation is updated via aggregating the representations of the corresponding (linked) hyperedges by a function $\psi(\cdot)$ (e.g., element-wise mean). 
The hypergraph representation $\bm{c}$ is obtained via further aggregating all the updated node representations by a function $\phi(\cdot)$ that is similar to $\psi(\cdot)$. Finally, we use a readout function $g(\cdot)$ to map $\bm{c}$ to a scalar value as the hypergraph classification result.
The IHGNN function $f_{IHGNN}$ is:
\begin{equation}
\label{fun:IHGNN}
\small
f_{IHGNN}(G_H\la \V, \E^{\prime}\ra) = g(\phi(\{\psi(\{\bm{h}_j|\bm{E}^{\prime}_{ij}=1\}_{j\in\E})\}_{i \in \V})).
\end{equation}

Inspired by \citet{su2020detecting}, we leverage an $L_0$ activation regularization to help detect beneficial feature interactions via sparsifying the generated hyperedge matrix $\bm{E}^{\prime}$. 
However, the $L_0$ activation regularization solely may not lead to a successful interaction detection of arbitrary orders. 
For example, it may lead to the generated incidence matrix containing only one edge that links to all nodes and the other $k-1$ edges being empty. As a result, all interaction information assembles in one hyperedge, which is not a beneficial feature interaction.
Therefore, we exploit three properties of beneficial feature interactions in the next section to guide the beneficial feature interaction generation.

\subsection{The s-Infomax and Infomin Component}
\label{sec:three_prop}

We first define three properties of beneficial feature interactions (i.e., beneficial properties). Then we describe how we achieve the beneficial properties by s-Infomax and Infomin methods.

\subsubsection{Beneficial Properties}

We formally define the three beneficial properties that beneficial feature interactions in a data sample should try to achieve:

\noindent
\underline{\textit{Sufficiency}}:
The representations of all beneficial feature interactions together contain as much information of the input about the true output as possible. In information theory, ideally, we aim to reach:
\begin{equation}
\label{fun:sufficiency}
I((\bh_1, \bh_2,...,\bh_k);y)=I(\bm{x};y),
\end{equation}
where $I(\cdot)$ is the mutual information.

\noindent
\underline{\textit{Low Redundancy}}:
Each beneficial feature interaction representation should contain as less irrelevant information about the true output as possible.
Ideally, we aim to reach:
\begin{equation}
\label{fun:low_redundancy}
\sum\nolimits_{i=1}^{k}H(\bm{h}_i|y)=0,
\end{equation}
where $H(\bh_i|y)$ is the entropy of $\bh_i$ conditioned on $y$. 

\noindent
\underline{\textit{Disentanglement}}:
The representations of any two beneficial feature interactions contain as less duplicated information about the true output as possible. Ideally, we aim to reach:
\begin{equation}
\label{fun:disentanglement}
\sum\nolimits_{i\neq j}I(\bh_i;\bh_j;y)=0.
\end{equation}

Figure \ref{fig:three_properties} shows the relations between feature interactions and the true output according to the three properties.

We aim to generate hyperedges that can simultaneously achieve the three properties.
In HIRS, all the hyperedge representations $\{\bh_1, \bh_2,...,\bh_k\}$ are generated by a function with features $\bm{x}$ as the input, and the output $y$ is predicted through the hyperedge representations, which forms a Markov chain $\bm{x}\rightarrow \{\bh_1, \bh_2,...,\bh_k\}\rightarrow y$. 
Therefore, $I(\bm{x};y) \geqslant I((\bh_1, \bh_2,...,\bh_k);y)$. 
To achieve the sufficiency property, we just need to maximize $I((\bh_1, \bh_2,...,\bh_k);y)$ to approach the optimal situation.

Due to the non-negativity of entropy and mutual information, $\sum_{i=1}^{k}H(\bh_i|y)\geqslant 0$ and $\sum_{i\neq j}I(\bh_i;\bh_j;y)\geqslant 0$. To achieve the low redundancy and the disentanglement properties, therefore, we minimize $\sum_{i=1}^{k}H(\bh_i|y)$ and $\sum_{i\neq j}I(\bh_i;\bh_j;y)$.

Combining all the goals together, the objective function is:
\begin{equation}
\label{fun:three_properties}
\small
\max I((\bh_1, \bh_2,...,\bh_k);y) - \alpha \sum_{i=1}^{k}H(\bh_i|y) - \beta\sum_{i\neq j}I(\bh_i;\bh_j;y),
\end{equation}
where $\alpha$ and $\beta$ are scalar weights balancing the three properties.

\begin{figure}[t]
\centering
\centerline{\includegraphics[width=0.9\columnwidth]{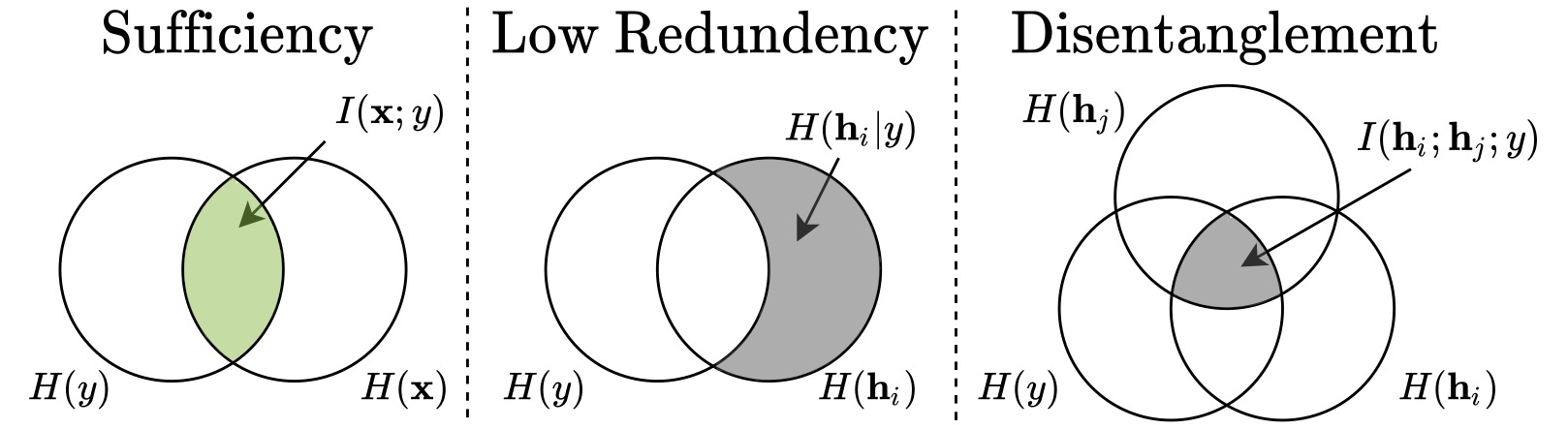}}
\caption{\small The Venn diagram of the three properties. The \textit{green} part is the information that our generated interactions to retain, and the \textit{gray} part is to discard.}  
\label{fig:three_properties}
\end{figure}

\subsubsection{s-Infomax and Infomin}

Due to the difficulty of directly optimizing Equation \eqref{fun:three_properties} \cite{hjelm2018learning}, 
we propose supervised mutual information maximization (s-Infomax) and mutual information minimization (Infomin) methods based on Deep Infomax (DIM, \citet{hjelm2018learning}) as the approximation of Equation \eqref{fun:three_properties}. 
The s-Infomax method maximizes the mutual information between each edge representation and the label representation, i.e., $\max I(\bh_i;y)$. The Infomin method minimizes the mutual information between every pair of edge representations in a hypergraph, i.e., $\min I(\bh_i;\bh_j)$. The approximation of Equation \eqref{fun:three_properties} is:
\begin{equation}
\label{fun:infomaxmin}
\small
\max \sum\nolimits_{i=1}^{k}I(\bh_i;y) - \sigma \sum\nolimits_{i\neq j}I(\bh_i;\bh_j),
\end{equation}
where $\sigma$ is a scalar weight.

We prove that Equation \eqref{fun:three_properties} reaches the optimal value when Equation \eqref{fun:infomaxmin} reaches the optimal value in Proposition \ref{prop:three_properties}. 
We detail the proof in Appendix \ref{appx:prop1_proof}.
\begin{proposition}
\label{prop:three_properties}
Consider an input-output pair $(\bm{x},y)$, where $\bm{x}$ is a set of features, a function maps $\bm{x}$ to $y$ in a Markov chain $\bm{x}\rightarrow \{\bh_1, \bh_2,...,\bh_k\}\rightarrow y$, where $\{\bh_1, \bh_2,...,\bh_k\}$ is a set of middle states mapped from $\bm{x}$, if the $\{\bh_1, \bh_2,...,\bh_k\}$ makes Equation \eqref{fun:infomaxmin} reach the optimal value, it also makes Equation \eqref{fun:three_properties} reach the optimal value.
\end{proposition}

Proposition \ref{prop:three_properties} ensures that by optimizing Equation \eqref{fun:infomaxmin}, we approximately optimize Equation \eqref{fun:three_properties} so that generate feature interactions that fit the three properties.

\subsubsection{The Implementation}

One difficulty of applying s-Infomax is that the label space only has two states, i.e., 0 and 1, which are discrete and are too small to provide enough output information \cite{khosla2020supervised}. 
Hence, we use the hypergraph representations to represent the label space, extending the two states into a continuous high-dimensional space. 
Following DIM, we sample the \textit{joint distribution} of $I(\bh_i;y)$ as the pair of $\bh_i$ and a random hypergraph representation $\bm{c}^{+}$ from training samples having the same label $y$. Then, we sample the \textit{marginal distribution} as $\bh_i$ and a random hypergraph representation $\bm{c}^{-}$ from training samples having the opposite label $y$. 
We use a GAN-based objective function for s-Infomax:
\begin{equation}
\label{fun:infomax_loss}
\small
\begin{split}
\max\sum\nolimits_{i}I(\bh_i;y) = &\max\frac{1}{k}\sum\nolimits_i (\log D_{\theta}(\bh_i, \bm{c}^{+}) \\
&+(1-\log D_{\theta}(\bh_i,\bm{c}^{-}))),
\end{split}
\end{equation}
where $D_{\theta}\in\R^{2\times d}\rightarrow\R$ is a discriminator that distinguishes the joint distribution and the marginal distribution.

In Infomin, we minimize the mutual information between different hyperedge representations in a hypergraph. We set the sampled joint distribution as $(\bh_i, \bh_j)$ that $i\neq j$, and the marginal distribution as $(\bh_i, \bh_i)$.
To avoid the overfitting issue that the discriminator in Infomin may trivially compare whether two input vectors are the same, inspired by \citet{gao2021simcse}, we use a dropout function $f_a$ to prevent each pair of input vectors from being the same.
 
To ensure efficiency, we get one sample from the joint distribution and one sample from the marginal distribution for each node. 
The objective function for Infomin is also a GAN-based objective function but swaps the joint and marginal distributions:
\begin{equation}
\label{fun:infomin_loss}
\small
\begin{split}
\min \sum_{i\neq j}I(\bh_i;\bh_j) = &\max \frac{1}{k}\sum_{i}(\log D_{\omega}(f_a(\bh_i) f_a(\bh_i)) \\
+&(1-\log D_{\omega}(f_a(\bh_i),f_a(\bh_j)))).
\end{split}
\end{equation}

Equation \eqref{fun:infomax_loss} and Equation \eqref{fun:infomin_loss} can be efficiently optimized by training the discriminators $D_{\theta}$ and $D_{\omega}$ through Binary Cross-entropy (BCE) loss \cite{velivckovic2018deep}.

Note that the main difference between s-Infomax and Infomin in implementation is the choice of the joint distribution and the marginal distribution that a discriminator to distinguish. 
The distributions chosen in Infomin are intuitive: we want any two hyperedge representations in a hypergraph to be as different as possible (i.e., have minimum mutual information).

\subsection{The Empirical Risk Minimization Function}
The empirical risk minimization function of HIRS contains the loss function of recommendation predictions, an $L_0$ activation regularization, the loss functions of the s-Infomax discriminator and the Infomin discriminator:
\begin{equation}
\label{fun:model_loss}
\small
\begin{split}
    \mathcal{R}(\bm{\theta}, \bm{\omega}, \bm{\rho})=&\frac{1}{N}\sum_{n=1}^{N}(\mathcal{L}(f_{r}((G_H)_n;\bm{\rho}),y_{n}) +\lambda_1 \norm{\bm{E}^{\prime}_{n}}_0 \\
    &+ \lambda_2 \mathcal{L}_D(D_{\theta}) + \lambda_3 \mathcal{L}_D(D_{\omega})), \\
    \bm{\rho}^{*}&,\bm{\theta}^{*}, \bm{\omega}^{*}=\argmin_{\bm{\theta},\bm{\omega},\bm{\rho}} {\mathcal{R}(\bm{\theta}, \bm{\omega},\bm{\rho})}, 
\end{split}
\end{equation}
where $(G_H)_n=(G_H\la \V,\emptyset\ra)_n$, $\lambda_1$, $\lambda_2$ and $\lambda_3$ are weight factors, $\mathcal{L}(\cdot)$ and $\mathcal{L}_D(\cdot)$ are both BCE loss, and $\bm{\rho}^{*}$, $\bm{\theta}^{*}$, $\bm{\omega}^{*}$ are final parameters. Following \citet{shi2019pvae}, the $L_0$ activation regularization can be calculated as: $\norm{\bm{E}^{\prime}_{n}}_0 = \sum_{i,j} Sig(\log u_{ij} - \tau \log \frac{-\gamma}{\delta}).$

\subsection{Generalizing IHGNN to Existing Methods}
We show that some existing methods can be represented as an IHGNN in specific settings.

\subsubsection{$L_0$-SIGN as IHGNN}
A hypergraph is an extension of a standard graph. Therefore, GNN-based recommender systems such as $L_0$-SIGN \cite{su2020detecting} and Fi-GNN \cite{li2019fi} can be easily applied by IHGNN by changing the edge matrix into the incidence matrix form and replacing the corresponding feature interaction modeling functions, e.g., dot product by an MLP for $L_0$-SIGN.

\subsubsection{FM as IHGNN}
FM \cite{rendle2010factorization} models every pairwise feature interaction by dot product and sums up all the modeling results. We can regard each pairwise feature interaction in FM as a hyperedge linking to two nodes,
which forms the $\E^{FM}$ that contains all pairwise node combinations, including all self-loops as the point-wise terms in FM.
In addition, FM-based extensions can also be achieved based on the above derivation. 
For example, NFM \cite{he2017neural} can be achieved by replacing the dot product and the $\phi$ by an element-wise production and an MLP, respectively.

\subsubsection{DeepFM as IHGNN}
DeepFM \cite{guo2017deepfm} merges the prediction of FM and the prediction of an MLP modeling all the features. We can regard DeepFM as IHGNN by adding a hyperedge linking to all nodes in addition to $\E^{FM}$. 
Since DeepFM uses different interaction modeling functions in FM and the MLP, the hyperedges are modeled differently according to the edge degree.

\subsection{Interaction Generation Cost Analysis}
\label{sec:time_complexity}

HIRS generates $k$ beneficial feature interactions through the hyperedge prediction module $f_{gen}$, which has the time complexity of $O(mkd)$ for a data sample of $m$ features, where $d$ is the embedding dimension.
The s-Infomax and the Infomin methods help the interaction generation and both have the complexity of $O(kd)$. 
In addition, the IHGNN module models the generated feature interactions for recommendation with a $O(mkd+md)$ complexity. 
Therefore, the time complexity of HIRS is $O(mkd)$.
Existing studies model or examine all the possible interactions by enumerating them $O(2^md)$ \cite{liu2020autofis} or stacking pairwise interaction modeling layers $O(m^3d)$ \cite{song2019autoint,wang2021dcn} when considering arbitrary orders.
Since $k$ can be set to be a much smaller number (e.g., 40) than $m^2$ (e.g., 196 when $m=14$), HIRS achieves a more efficient interaction modeling.
We do not include the complexity of the neural network modules in HIRS and the baselines as they are all linear to $d$.
The running time of HIRS and the baselines are detailed in Appendix \ref{appx:time_complexity}. 

\section{Experimental Analysis}
\label{sec:experiences}

In this section, we evaluate the following aspects of our proposed model: 
(i) the recommendation prediction accuracy of HIRS; 
(ii) the beneficial feature interactions generation ability of HIRS; 
(iii) the impact of different components and hyperparameters of HIRS.

\subsection{Experimental Setting}

\subsubsection{Datasets and Baselines}

\begin{table}[t]
\centering
\begin{tabular}{lcccc}
\hline
\textbf{Dataset} & \textbf{\#Data} & \textbf{\#Users} & \textbf{\#Items} &\textbf{\#Features}\\
\hline
MovieLens 1M     & 1,149,238 & 5,950  &  3,514 &  6,974 \\
Book-Crossing & 1,050,834 & 4,873  & 53,168 & 43,244 \\
MovieLens 25M   &  31,147,118 &  162,541  &  62,423 &  226,122 \\
\hline
\end{tabular}
\caption{\small Dataset Statistics.}
\label{tab:dataset}
\vskip -0.24in
\end{table}

We evaluate HIRS on three real-world datasets in recommender systems: MovieLens 1M \cite{harper2015movielens}, Book-crossing \cite{ziegler2005improving}, and MovieLens 25M \cite{harper2015movielens}. 
We summarize the statistics of the datasets in Table \ref{tab:dataset}.
Note that although MovieLens 1M and MovieLens 25M are both movie recommendation datasets, they are actually deployed with different features, e.g., MovieLens 25M contains the tag features, while MovieLens 1M does not. In this way, we can evaluate the effectiveness and the robustness of our model when different kinds of features are available in similar scenarios. 
We compare HIRS with state-of-the-art models considering feature interactions for recommendation.
They are
FM \cite{koren2008factorization},
AFM \cite{xiao2017attentional},
NFM \cite{he2017neural},
Fi-GNN \cite{li2019fi},
$L_0$-SIGN \cite{su2020detecting},
AutoInt \cite{song2019autoint},
DeepFM \cite{guo2017deepfm},
xDeepFM \cite{lian2018xdeepfm},
DCNv2 \cite{wang2021dcn},
AutoFIS \cite{liu2020autofis},
AFN \cite{cheng2020adaptive}.
More details about the datasets and the baselines are in Appendix \ref{appx:dataset_baseline}.

\begin{table*}[ht]
\centering
\begin{tabular}{l|c|c|cccc|cccc|cccc}
\hline
 &\textbf{HO}&\textbf{BD}& \multicolumn{4}{c|}{\textbf{MovieLens 1M}} & \multicolumn{4}{c|}{\textbf{Book-Crossing}} & \multicolumn{4}{c}{\textbf{MovieLens 25M}}\\
 && & R@10 & R@20 & N@10 & N@20 & R@10 & R@20 & N@10 & N@20 & R@10 & R@20 & N@10 & N@20\\
\hline
FM          &       &       & 0.8527 & 0.8996 & 0.8521 & 0.8845 & 0.7802 & 0.8834 & 0.8123 & 0.8539 & 0.8331 & 0.8972 & 0.9535 & 0.8791 \\
AFM         &       &       & 0.8738 & 0.9176 & 0.8734 & 0.8994 & 0.8084 & 0.8892 & 0.8473 & 0.8779 & 0.8484 & 0.9082 & 0.8618 & 0.8845 \\
NFM         &       &       & 0.8979 & 0.9318 & 0.8925 & 0.9084 & 0.8625 & 0.9174 & 0.8492 & 0.8843 & 0.8604 & 0.9170 & 0.8620 & 0.8903\\			
\hline
Fi-GNN      &       &       & 0.9060 & 0.9341 & 0.9091 & 0.9226 & 0.8818 & 0.9262 & 0.8689 & 0.8861 & 0.8692 & 0.9124 & 0.8667 & 0.8905 \\
$L_0$-SIGN  &       &$\surd$& \underline{0.9092} & \underline{0.9386} & \underline{0.9176} & \underline{0.9268} & 0.8817 & 0.9266 & \underline{0.8723} & 0.8914 & 0.8729 & 0.9261 & 0.8704 & 0.8915 \\
\hline
AutoInt     &$\surd$&       & 0.8991 & 0.9318 & 0.9077 & 0.9118 & 0.8739 & 0.9263 & 0.8636 & 0.8894 & 0.8707 & 0.9233 & 0.8692 & 0.8899 \\
DeepFM      &$\surd$&       & 0.8975 & 0.9308 & 0.9074 & 0.9128 & 0.8765 & 0.9245 & 0.8648 & 0.8800 & 0.8667 & 0.9213 & 0.8650 & 0.8883 \\
xDeepFM     &$\surd$&       & 0.9049 & 0.9342 & 0.9134 & 0.9114 & 0.8791 & 0.9249 & 0.8692 & 0.8826 & 0.8673 & 0.9226 & 0.8696 & 0.8902 \\
DCNv2      &$\surd$&       & 0.9014 & 0.9311 & 0.9079 & 0.9104 & 0.8773 & 0.9234 & 0.8671 & 0.8858 & 0.8678 & 0.9214 & 0.8667 & 0.8890 \\
AutoFIS     &$\surd$&$\surd$& 0.9028 & 0.9317 & 0.9096 & 0.9208 & 0.8804 & 0.9227 & 0.8706 & 0.8887 & 0.8702 & 0.9228 & 0.8712 & 0.8922 \\
AFN         &$\surd$&$\surd$& 0.9032 & 0.9333 & 0.9112 & 0.9232 & \underline{0.8821} & \underline{0.9285} & 0.8710 & \underline{0.8917} & \underline{0.8743} & \underline{0.9271} & \underline{0.8737} & \underline{0.8938} \\
\hline
HIRS       &$\surd$&$\surd$& \textbf{0.9545} & \textbf{0.9663} & \textbf{0.9464} & \textbf{0.9500} & \textbf{0.9279} & \textbf{0.9584} & \textbf{0.9081} & \textbf{0.9211} & \textbf{0.9223} & \textbf{0.9545} & \textbf{0.8846} & \textbf{0.9011} \\
\hline
\textit{Improv.}  &  &  & 4.98\% & 2.95\% & 3.13\% & 2.50\% & 5.19\% & 3.22\% & 4.10\% & 3.30\% & 5.49\% & 2.96\% & 1.26\% & 0.81\%\\
\textit{p-value}  &  &  & 0.5\%  & 0.5\%  & 0.5\%  & 0.5\%  & 0.5\%  & 0.5\%  & 0.5\%  & 0.5\%  & 0.5\% & 0.5\% &  1.4\% & 3.7\% \\       
\hline
\end{tabular}
\caption{\small Comparing the prediction performance of HIRS with the baselines. \textbf{HO} and \textbf{BD} indicate whether the model can consider high-order feature interactions and perform beneficial feature interaction detection, respectively. R@j refers to Recall, and N@j refers to NDCG.}
\label{tab:performance}
\vskip -0.1in
\end{table*}

\subsubsection{Experimental Setup}
In the experiments, we use element-wise mean as the linear aggregation function for both $\psi(\cdot)$ and $\phi(\cdot)$. $g(\cdot)$ is a linear regression function.
The MLPs for $f_{gen}(\cdot)$ and $f_{E}$ consist of a fully connected layer of size 64 with a ReLU activation. 
Both $D_{\theta}$ and $D_{\omega}$ are Bilinear models.
The dropout probability of $f_a$ is 0.1.
The node embedding size for interaction modeling and edge prediction is 64.
The hyperedge number $k$ is 40. 
We set $\lambda_1=0.02$, $\lambda_2=1$, $\lambda_3=0.1$, and Adam \cite{kingma2014adam} as the optimization algorithm.
We randomly split each dataset into training, validation, and test sets with a proportion of 70\%, 15\%, and 15\%, respectively.  
We evaluate the performance of HIRS and the baseline models using the popular ranking metrics Recall@10, Recall@20, NDCG@10, and NDCG@20.

\subsection{Model Performance}
Table \ref{tab:performance} shows the performance of our model and the baseline models. The results are the average of 10 runs. The best results are in bold, and the best baseline results are underlined. 
The rows \textit{Improv} (standing for Improvements) and \textit{p-value} show the improvement and statistical significance test results (through Wilcoxon signed-rank test \cite{wilcoxon1992individual}) of HIRS to the best baseline results, respectively.

We observe that:
(i) Our model outperforms all the baselines on the three datasets (up to 5.49\%), and the significance test results are all lower than the threshold of 5\%, which show that the improvements are significant. 
(ii) For the models that only consider pairwise feature interactions, GNN-based models (e.g., Fi-GNN and $L_0$-SIGN) perform better than other baselines (e.g., NFM). This shows the ability of GNNs in interaction modeling.
(iii) Models considering higher-order feature interactions (e.g., xDeepFM and AFN) perform better than those only consider pairwise interactions (e.g., NFM and Fi-GNN). Thus, considering high-order feature interactions is beneficial for providing better prediction accuracy.
(iv) Detecting beneficial feature interactions is vital in removing noise interactions and improving prediction accuracy (e.g., $L_0$-SIGN v.s. Fi-GNN).
(v) Our model combines GNN-based interaction modeling, higher-order interaction modeling, and beneficial feature interaction detection, and achieves the best performance.

\begin{figure}[t]
\centering
\begin{tikzpicture}
\begin{axis}[barstyle, width=0.5\textwidth, height=3.3cm, legend columns=3, ylabel={Percent}, xlabel={Feature Interaction Order}, symbolic x coords = {0,1,2,3,4,5,6,7,8,9,10,11,12,13,14}, legend style={draw=none, at={(0.25,0.85)},anchor=west, nodes={scale=1, transform shape}, legend image post style={scale=0.7}},ymin=0, ymax=0.85]
\addplot [y filter/.expression={y==0 ? nan : y},color=mycolor4, fill=mycolor4] coordinates { 
(0, 0.12)
(1, 0.18)
(2, 0.14)
(3, 0.10)
(4, 0.04)
(5, 0.02)
(6, 0.01)
(7, 0.02)
(8, 0.04)
(9, 0.07)
(10, 0.07)
(11, 0.06)
(12, 0.05)
(13, 0.05)
(14, 0.03)
};
\addplot[y filter/.expression={y==0 ? nan : y},color=mycolor1, fill=mycolor1] coordinates { 
(0, 0.84)
(1, 0.01)
(2, 0.00)
(3, 0.00)
(4, 0.00)
(5, 0.00)
(6, 0.00)
(7, 0.00)
(8, 0.00)
(9, 0.00)
(10, 0.00)
(11, 0.00)
(12, 0.02)
(13, 0.04)
(14, 0.09)
};
\addplot[y filter/.expression={y==0 ? nan : y}, color=mycolor2, fill=mycolor2] coordinates { 
(0, 0.00)
(1, 0.00)
(2, 0.00)
(3, 0.00)
(4, 0.00)
(5, 0.00)
(6, 0.00)
(7, 0.00)
(8, 0.01)
(9, 0.04)
(10, 0.13)
(11, 0.24)
(12, 0.29)
(13, 0.21)
(14, 0.08)
};
\legend{HIRS, w/o MI, w/o $L_0$}
\end{axis}
\end{tikzpicture}
\vskip -0.10in
\caption{\small Comparing the distribution of the orders of the generated feature interactions from HIRS, HIRS without s-Infomax and Infomin, and HIRS without $L_0$ regularization for MovieLens 1M. }
\label{fig:influence_infomax_edgeprediction}
\vskip -0.10in
\end{figure}
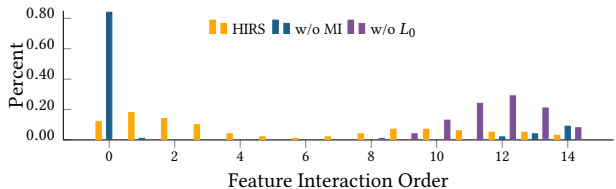

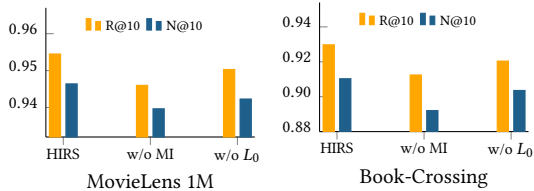
\begin{figure}[t]
\centering
\begin{tikzpicture}
\begin{axis}[barstyle, xlabel={MovieLens 1M},bar width=1.5mm, symbolic x coords = {HIRS, w/o MI, w/o $L_0$}, legend style={draw=none, at={(0.15,0.85)},anchor=west, nodes={scale=0.9, transform shape}, legend image post style={scale=0.6}},ymin=0.932, ymax=0.967]
\addplot [color=mycolor4, fill=mycolor4] coordinates { 
(HIRS, 0.9545)
(w/o MI, 0.9460)
(w/o $L_0$, 0.9503)
};
\addplot[color=mycolor1, fill=mycolor1] coordinates { 
(HIRS, 0.9464)
(w/o MI, 0.93965)  
(w/o $L_0$,  0.9423) 
};
\legend{R@10, N@10}
\end{axis}
\end{tikzpicture}
\begin{tikzpicture}
\begin{axis}[barstyle, xlabel={Book-Crossing},bar width=1.5mm, symbolic x coords = {HIRS, w/o MI, w/o $L_0$}, legend style={draw=none, at={(0.15,0.85)},anchor=west, nodes={scale=0.9, transform shape}, legend image post style={scale=0.6}},ymin=0.88, ymax=0.954]
\addplot [color=mycolor4, fill=mycolor4] coordinates { 
(HIRS, 0.929728)
(w/o MI, 0.91245)
(w/o $L_0$, 0.920415)
};
\addplot[color=mycolor1, fill=mycolor1] coordinates { 
(HIRS, 0.910294)         
(w/o MI, 0.89205)  
(w/o $L_0$,  0.903591) 
};
\legend{R@10, N@10}
\end{axis}
\end{tikzpicture}
\vskip -0.10in
\caption{\small Comparing the performance of HIRS, HIRS without s-Infomax and Infomin, and HIRS without $L_0$ regularization.}
\label{fig:influence_infomax_performance}
\end{figure}

\subsection{Effectiveness of the Beneficial Properties}
\label{subsec:inter_detec_eval}
To ensure our model can detect beneficial feature interactions that fit the beneficial properties (sufficiency, low-redundancy, and disentanglement), we propose the s-Infomax and Infomin module and the $L_0$ activation regularization.
This section, we evaluate the influence of them on the beneficial feature interaction detection and the prediction accuracy.

We run two variants of HIRS: (i) without s-Infomax and Infomin (\textit{w/o MI}), and (ii) without the $L_0$ activation regularization (\textit{w/o $L_0$}). 
Figure \ref{fig:influence_infomax_edgeprediction} shows the distribution of the generated feature interactions on three settings (the full model and the two variants).
We observe that when using the variant \textit{w/o MI},
most of the generated hyperedges are empty (order 0), and a few hyperedges connect to most features (order 12-14).
It proves the importance of the three properties on arbitrary order beneficial interaction generation. 
When using the variant \textit{w/o $L_0$}, most of the generated interactions have relative higher orders than using our full model. This proves that $L_0$ regularization helps ease the issue of mixed interaction in each hyperedge. 
Next, we show the prediction accuracies of the two variants and the full model in Figure \ref{fig:influence_infomax_performance}.
We show the results of Recall@10 and NDCG@10. Similar conclusions can be drawn from Recall@20 and NDCG@20 (the same as the remaining experiments).
We can see that both \textit{w/o MI} and \textit{w/o $L_0$} reduce the prediction accuracy. 
This proves that our full model successfully generates more beneficial feature interactions than the two variants and hence gains the best performance.

\subsection{Ablation Study on RPC Modules}

We evaluate the impact of two key functions, the hyperedge detection function $f_{gen}$ and the non-linear interaction modeling function $f_E$, in the two modules (the hyperedge prediction module and the IHGNN module) of the recommendation prediction component (RPC), respectively.
To evaluate the impact of $f_{gen}$, we remove $f_{gen}$ and use a fully connected matrix as the hyperedge matrix, i.e., $\bm{E}^{\prime}=\bm{1}$. 
We name this variant as \emph{w/o HP}.
To evaluate $f_E$ in HIRS, we replace $f_E$ by a linear function, which simply remove the ReLU activation in $f_E$.
We name this variant as \emph{w/o NM}.

We summarize the results in Figure \ref{fig:influence_component_performance}. 
The variant \emph{w/o HP} achieves an inferior performance.
In this situation, the hyperedge does not indicate a feature interaction, but all the features. 
Therefore, the interaction modeling function does not model beneficial feature interactions but the mix of all the feature information, which inevitably contains the noise information that decreases the prediction accuracy.
The variant \emph{w/o NM} also achieves an inferior performance. This is because the non-linear interaction modeling helps correctly model the generated feature interactions and hence improves the accuracy. This is consistent with the claim in \citet{su2020detecting}.
We also show the results of the variant with both modifications (\emph{w/o Both}). The worst performance it achieves is not surprising since the whole model under this setting can be regarded as a linear mapping from feature embeddings to the final prediction.

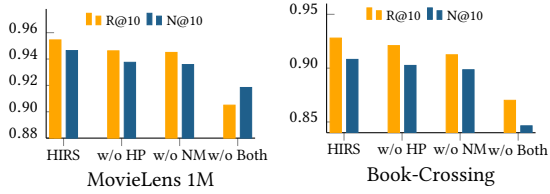
\begin{figure}[t]
\centering
\begin{tikzpicture}
\begin{axis}[barstyle, xlabel={MovieLens 1M},bar width=1.5mm, symbolic x coords = {HIRS, w/o HP, w/o NM, w/o Both}, legend style={draw=none, at={(0.15,0.90)},anchor=west, nodes={scale=0.9, transform shape}, legend image post style={scale=0.6}},ymin=0.88, ymax=0.978]
\addplot [color=mycolor4, fill=mycolor4] coordinates { 
(HIRS,     0.9545)
(w/o HP,  0.9462)
(w/o NM,  0.9450)
(w/o Both,0.9050)
};
\addplot[color=mycolor1, fill=mycolor1] coordinates { 
(HIRS,     0.9464)
(w/o HP,  0.9375)
(w/o NM,  0.9358)
(w/o Both,0.9184)
};
\legend{R@10, N@10}
\end{axis}
\end{tikzpicture}
\begin{tikzpicture}
\begin{axis}[barstyle, xlabel={Book-Crossing},bar width=1.5mm, symbolic x coords = {HIRS, w/o HP, w/o NM, w/o Both}, legend style={draw=none, at={(0.15,0.90)},anchor=west, nodes={scale=0.9, transform shape}, legend image post style={scale=0.6}},ymin=0.84, ymax=0.96]
\addplot [color=mycolor4, fill=mycolor4] coordinates { 
(HIRS,     0.9279)
(w/o HP,  0.9209)
(w/o NM,  0.9124)
(w/o Both,0.8701)
};
\addplot[color=mycolor1, fill=mycolor1] coordinates { 
(HIRS,     0.9081)
(w/o HP,  0.9025)
(w/o NM,  0.8985)
(w/o Both,0.8464)
};
\legend{R@10, N@10}
\end{axis}
\end{tikzpicture}
\vskip -0.1in
\caption{\small Comparing the recommendation performance of HIRS, HIRS without hyperedge prediction, HIRS without non-linear interaction modeling, and HIRS without both.}
\label{fig:influence_component_performance}
\vskip -0.1in
\end{figure}

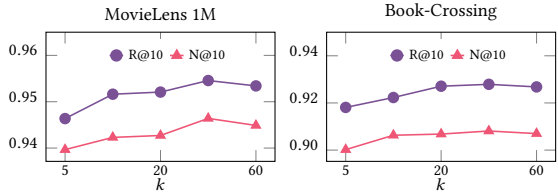
\begin{figure}[t]
\centering
\begin{tikzpicture}
\begin{axis}[linestyle, title =MovieLens 1M, xlabel={$k$},legend columns=2, ymax=0.965, symbolic x coords={5, 10, 20, 40, 60}, legend style={draw=none, at={(0.84,0.82)},anchor=east, nodes={scale=0.9, transform shape}}, legend image post style={scale=0.6}]
\addplot[mark=*, color=mycolor2] coordinates {
(5, 0.94636 )
(10, 0.951621 )
(20, 0.952079 )
(40, 0.95457 )
(60, 0.95341 )
};
\addplot[mark=triangle*,color=mycolor3] coordinates {
(5, 0.939662)
(10, 0.942298)
(20, 0.942713)
(40, 0.94640)
(60, 0.94488)
};
\legend{R@10, N@10}
\end{axis}
\end{tikzpicture}
\begin{tikzpicture}
\begin{axis}[linestyle, title=Book-Crossing, xlabel={$k$},legend columns=2, ymax=0.950, symbolic x coords={5, 10, 20, 40,60}, legend style={draw=none, at={(0.84,0.82)},anchor=east, nodes={scale=0.9, transform shape}}, legend image post style={scale=0.6}]
\addplot[mark=*, color=mycolor2] coordinates {
(5,  0.9181)
(10, 0.9223)
(20, 0.9271)
(40, 0.9279)
(60, 0.9268)
};
\addplot[mark=triangle*,color=mycolor3] coordinates {
(5,  0.9002)
(10, 0.9063)
(20, 0.9068)
(40, 0.9081)
(60, 0.9070)
};
\legend{R@10, N@10}
\end{axis}
\end{tikzpicture}
\vskip -0.1in
\caption{\small Comparing the recommendation performance of HIRS with different hyperedge numbers.}
\label{fig:para_study_k}
\vskip -0.1in
\end{figure}

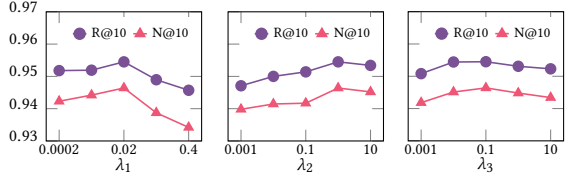
\begin{figure}[t]
\centering
\begin{tikzpicture}
\begin{axis}[linestyle,width=0.205\textwidth, xlabel={$\lambda_1$},legend columns=2, ymax=0.970, ymin=0.930, symbolic x coords={0.0002, 0.002, 0.02, 0.2, 0.4}, legend style={draw=none, at={(1,0.830)},anchor=east, nodes={scale=0.9, transform shape}}, legend image post style={scale=0.6}]
\addplot[mark=*, color=mycolor2] coordinates {
(0.0002,  0.95179)
(0.002, 0.95192)
(0.02,  0.95450 )
(0.2,  0.94896 )
(0.4, 0.94570)
};
\addplot[mark=triangle*,color=mycolor3] coordinates {
(0.0002, 0.9423 )
(0.002, 0.9442)
(0.02, 0.94641 )
(0.2, 0.93872 )
(0.4, 0.93420 )
};
\legend{R@10, N@10}
\end{axis}
\end{tikzpicture}
\begin{tikzpicture}
\begin{axis}[linestyle, width=0.205\textwidth, yticklabels={,,}, xlabel={$\lambda_2$},legend columns=2, ymax=0.970,ymin=0.930, symbolic x coords={0.001, 0.01, 0.1, 1, 10}, legend style={draw=none, at={(1,0.830)},anchor=east, nodes={scale=0.9, transform shape}}, legend image post style={scale=0.6}]
\addplot[mark=*, color=mycolor2] coordinates {
(0.001, 0.94709)
(0.01,  0.9500)
(0.1,   0.9514)
(1,     0.9545)
(10,    0.9534)
};
\addplot[mark=triangle*,color=mycolor3] coordinates {
(0.001, 0.93981)
(0.01,  0.94145)
(0.1,   0.94170)
(1,     0.9464)
(10,    0.94519)
};
\legend{R@10, N@10}
\end{axis}
\end{tikzpicture}
\begin{tikzpicture}
\begin{axis}[linestyle, width=0.205\textwidth, yticklabels={,,}, xlabel={$\lambda_3$},legend columns=2, ymax=0.970,ymin=0.930, symbolic x coords={0.001, 0.01, 0.1, 1, 10}, legend style={draw=none, at={(1,0.830)},anchor=east, nodes={scale=0.9, transform shape}}, legend image post style={scale=0.6}]
\addplot[mark=*, color=mycolor2] coordinates {
(0.001, 0.950808 )
(0.01,  0.954417 )
(0.1, 0.9545)
(1,  0.95309 )
(10, 0.952301 )
};
\addplot[mark=triangle*,color=mycolor3] coordinates {
(0.001, 0.94183)
(0.01,  0.945116 )
(0.1, 0.9464)
(1, 0.9448 )
(10,0.94340 )
};
\legend{R@10, N@10}
\end{axis}
\end{tikzpicture}
\vskip -0.1in
\caption{\small Comparing the performance of different weight values of the $L_0$ activation regularization, s-Infomax, and Infomin for MovieLens 1M.}
\label{fig:para_study_weights}
\vskip -0.1in
\end{figure}

\subsection{Parameter Study}
We evaluate the impact of the number of feature interactions ($k$), the weight values for $L_0$ activation regularization ($\lambda_1$), s-Infomax ($\lambda_2$), and Infomin ($\lambda_3$) in Equation \eqref{fun:model_loss} on the prediction accuracy. 
Figure \ref{fig:para_study_k} shows the performance of HIRS using different numbers of edges ($k$).
The prediction accuracy increases when the $k$ increases from 5 to 40, but flattens over 40. 
This is because that when the edge number is small, some of the beneficial  feature interactions are discarded or mixed in one hyperedge. 
However, the number of edge slots is enough to contain beneficial feature interactions when $k>40$.
Next, we evaluate different $\lambda_1, \lambda_2, \lambda_3$ values in Equation \eqref{fun:model_loss} on the MovieLens 1m dataset. 
As shown in Figure \ref{fig:para_study_weights},
every weight has a sweat point that leads to the best performance. 
For the $L_0$ regularization weight ($\lambda_1$), it balances the sparsity and the useful information retraining of the beneficial feature interactions. 
s-Infomax and Infomin balance the three properties. The sufficiency and the low redundancy are achieved by s-Infomax, and the disentanglement is achieved by Infomin. For example, a high s-Infomax weight ensures sufficiency and low redundancy, but may harm the disentanglement. Only when the three properties achieve simultaneously we can successfully generate beneficial feature interactions, and leverage them for accurate recommendations.

\subsection{Case Study on the Detected Interactions}

We do a case study to visualize the detected feature interactions. Figure \ref{fig:case_study} shows the detected feature interactions (in an incidence matrix) of a data sample in the MovieLens 1M dataset, where empty ones were removed. Each column is a feature interaction with corresponding features colored. We sort the feature interactions by order and set them with different colors based on their order, e.g., dark blue for second order feature interactions and purple for 13 order feature interactions (the number of colored cells in each column indicates the order of the corresponding interaction).
From the figure, we observe that: 
(i) Feature interactions provide potential explanations for the recommendation prediction results. For example, interaction 13 shows the interaction between age and comedy genre, which indicates that users aged 25-35 may prefer comedy movies. 
(ii) Features such as gender and age participate in many feature interactions, which shows the importance of these features in the prediction.
(iii) User id (\textit{e.g., User\_212}) and movie id (e.g., \textit{Naked Gun}) usually appears in high order feature interactions (e.g., interaction 30-34). This indicates that the id features can be used to infer complex preferences or properties.
(iv) There are a few duplicated feature interactions, e.g., interactions 13 and 17. It means that these interactions are important. Having them duplicated adds the weights of these interactions in recommendation predictions.

Due to the space limitation, we present some experimental results for the Book-Crossing dataset and the MovieLens 25M dataset in Appendix \ref{appx:additional_exp}, which are consistent with demonstrated results in the above sections.

\begin{figure}[t]
\centerline{\includegraphics[width=0.98\columnwidth]{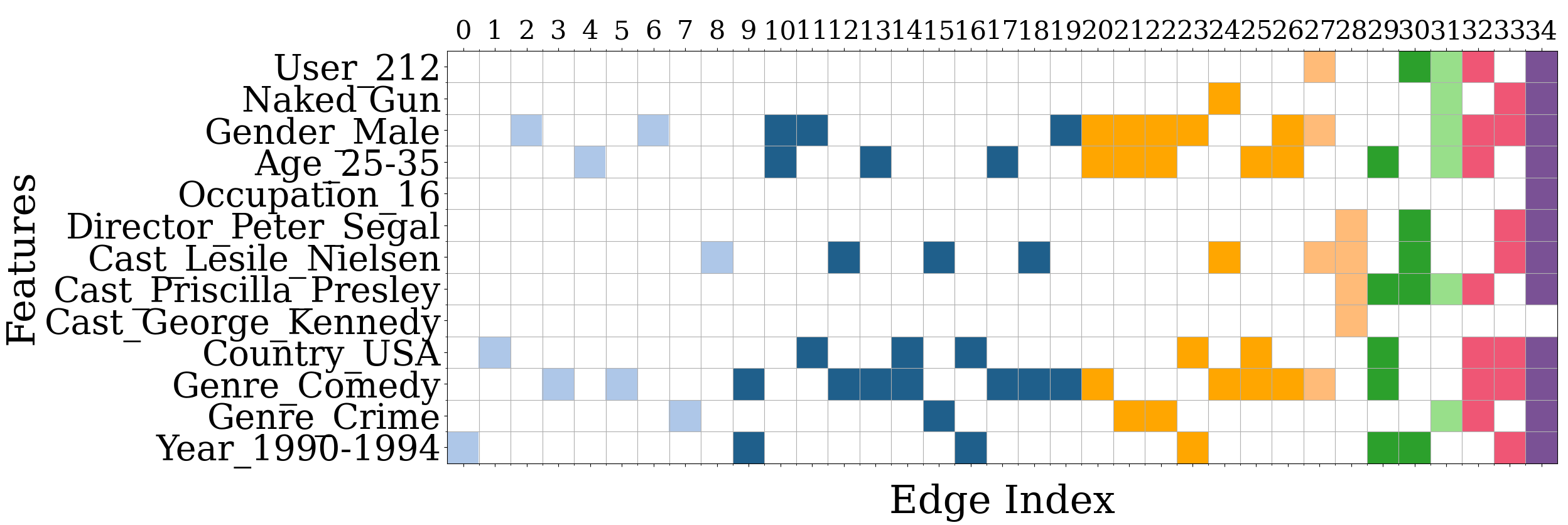}}
\vskip -0.1in
\caption{\small The detected feature interactions (in an incidence matrix) of a data sample in the MovieLens 1M dataset. Feature interactions are sorted by their orders and columns with the same color represent the interactions of the same order.}
\label{fig:case_study}
\end{figure}

\section{Conclusion}
\label{sec:conclusion}
We proposed a method for detecting arbitrary order beneficial feature interactions for recommendation.
The method represents each data sample as a hypergraph and uses a hypergraph neural network-based model named HIRS for recommendation.
HIRS achieves superior results by directly generating beneficial feature interactions via hyperedge prediction and performing recommendation predictions via hypergraph classification.
We define three properties of beneficial feature interactions, and propose s-Infomax and Infomin methods for effective interaction generation.
In the future work, we will explore the possibility to automatically determine the best edge number, and to extend s-Infomax and Infomin methods to other domains such as anomaly detection.

\begin{acks}
This work is supported by the China Scholarship Council (CSC).
In this research, Junhao Gan is in part supported by Australian Research Council (ARC)
Discovery Early Career Researcher Award (DECRA) DE190101118, and Sarah Erfani is in part supported by Australian Research Council (ARC) Discovery Early Career Researcher Award (DECRA) DE220100680.
\end{acks}

\balance
\bibliography{reference}
\bibliographystyle{ACM-Reference-Format}

\appendix

\section{Proof of Proposition 1}
\label{appx:prop1_proof}

\textbf{Proposition \ref{prop:three_properties}.}
\textit{
Consider an input-output pair $(\bm{x},y)$, where $\bm{x}$ is a set of features, a function maps $\bm{x}$ to $y$ in a Markov chain $\bm{x}\rightarrow \{\bh_1, \bh_2,...,\bh_k\}\rightarrow y$, where $\{\bh_1, \bh_2,...,\bh_k\}$ is a set of middle states mapped from $\bm{x}$, if the $\{\bh_1, \bh_2,...,\bh_k\}$ makes Equation \eqref{fun:infomaxmin} reach the optimal value, it also makes Equation \eqref{fun:three_properties} reach the optimal value.
}

\begin{proof}
Due to the non-negativity of mutual information, we have $I(\bh_i;\bh_j)\geqslant 0$. $\sum_{i\neq j}I(\bh_i;\bh_j)$ reaches the smallest value when $I(\bh_i;\bh_j)=0$ for all $i\neq j$. That is, each pair of hyperedge representations are independent. Therefore, we have
\begin{equation}
\label{fun:proof_1}
\small
I((\bh_1, \bh_2,...,\bh_k);y)=\sum_{i=1}^{k}I(\bh_i;y).
\end{equation}

It means that when $\sum_{i\neq j}I(\bh_i;\bh_j)$ reaches the smallest value, maximizing $\sum_{i=1}^{k}I(\bh_i;y)$ equals to maximizing $I((\bh_1, \bh_2,...,\bh_k);y)$. 
Next, since
\begin{equation}
\label{fun:proof_2}
\small
0 \leqslant I(\bh_i;\bh_j;y)\leqslant I(\bh_i;\bh_j),
\end{equation}
$\sum_{i\neq j}I(\bh_i;\bh_j;y)=0$ when $\sum_{i\neq j}I(\bh_i;\bh_j)=0$.

Finally, since $I(\bh_i;y)=H(\bh_i) - H(\bh_i|y)$, maximizing $I(\bh_i;y)$ tries to reduce $H(\bh_i|y)$ and in the optimal situation, $H(\bh_i|y)=0$.

In summary, when the formulation in Equation \eqref{fun:infomaxmin} reaches the optimal situation, we have maximized $I((\bh_1, \bh_2,...,\bh_k);y)$, and have $I(\bh_i;\bh_j;y) =0$ and $H(\bh_i|y)=0$, which make Equation \eqref{fun:three_properties} reach the optimal situation.
\end{proof}

\section{IHGNN Satisfying the Statistical Interaction Principle}
\label{appx:prop2_proof}

We first give the definition of high-order statistical interaction following \citet{sorokina2008detecting}.
\begin{definition}
\label{def:ho-si}
(\textbf{Statistical Interaction}) 
Consider a function $f(\cdot)$ with input variables $\bm{x}=\{x_1,x_2, ..., x_m\}$. Then, a set of variables $\I\subseteq \bm{x}$ is a statistical interaction of function $f(\cdot)$ if and only if there does \textbf{not} exist a set of functions $f_{\char`\\i}(\cdot)$, where $x_i\in\I$ and $f_{\char`\\i}(\cdot)$ is not a function of $x_i$, such that 
\begin{equation}
\label{fun:ho-si}
\small
f(\bm{x}) = \sum_{x_i\in\I}f_{\char`\\i}(x_1, x_{i-1}, x_{i+1}, ...,x_m)
\end{equation}
\end{definition}

Then, we prove that IGHNN models a set of features as a statistical interaction if they form a hyperedge in the input hypergraph in Proposition \ref{prop:si_in_HIRS}. 

\begin{proposition}
\label{prop:si_in_HIRS}
Consider an input-output pair $(\bm{x}, y)$. Let $G\la\V,\emptyset\ra$ be the input of HIRS, where $\V=\bm{x}$,
the IHGNN function (Equation \eqref{fun:IHGNN}) flags a statistical interaction among all the features in $q\subseteq \V$ if $q=\{a_i|\bm{E}^{\prime}_{ik}=1\}_{i\in \V}$, where $\bm{E}^{\prime}_k$ is a hyperedge predicted by $f_{gen}$ (Equation \eqref{fun:CrossRest}).
\end{proposition}

\begin{proof}

We prove by contradiction: assume there is a $q=\{c_i|\bm{E}^{\prime}_{ik}=1\}_{i\in \V}$ but the IHGNN function does not flag a statistical interaction on all the features in $q$.

In HIRS, the edge modeling function $f_{E}$ is a neural network, which is a non-linear function. Without losing generality, we set both the aggregation function $\phi$ and $\psi$ as the element-wise mean function and $g$ is a linear regression function. Therefore, we can regard $f_{IGHNN}$ as:
\begin{equation}
\label{fun:IHGNN_re}
\small
f_{IHGNN}(G_H\la \V, \E^{\prime}\ra) = \sum_{j\in\E^{\prime}}\alpha_j f_{E}(sum(\{\bm{V}^{g}_i|\bm{E}^{\prime}_{ij}=1\}_{i\in\V})),
\end{equation}
where $\alpha_j$ are some scalar values.

According to Definition \ref{def:ho-si}, if  $f_{IHGNN}$ does not flag a statistical interaction on all the features in $q$, we can rewrite Equation \eqref{fun:IHGNN_re} as a linear combination of a set of functions that each function does not rely on at least one of the features in $q$. That is:
\begin{equation}
\label{fun:IHGNN_assumption}
\small
f_{IHGNN}(G_H\la \V, \E^{\prime}\ra) = \sum_{i\in q}f_{\char`\\i}(\{\bm{V}^{g}_j\}_{j\in \V{\char`\\\{i\}}}),
\end{equation}

However, since all features in $q$ form a hyperedge in $\E^{\prime}$, Equation \eqref{fun:IHGNN_re} contains a component $f_{E}(sum(\{\bm{V}^{g}_i|i\in q\}_{i\in\V}))$, which non-linearly models all features in $q$. This component does not belong to any part of the right hand side of Equation \eqref{fun:IHGNN_assumption}. Therefore, Equation \eqref{fun:IHGNN_assumption} does not hold, which is a contradiction of the assumption. 
\end{proof}

Note that according to the definition of statistical interaction, if $f_{IHGNN}$ flags a statistical interaction on all features in $q$, $f_{IHGNN}$ flags a statistical interaction on any $q_{s}\subseteq q$. This is the reason that we use Infomin to force each hyperedge to contain different feature interactions and prevent multiple feature interactions from merging into one hyperedge.

\section{Datasets and Baselines}
\label{appx:dataset_baseline}

\paragraph{Datasets}
We study three real-world datasets to evaluate our model: \textbf{MovieLens 1M} \cite{harper2015movielens} contains users' ratings on movies. 
Each data sample contains a user and a movie with their corresponding attributes as features.
We further collect movies' other features, such as directors and casts from IMDB
to enrich the datasets.
\textbf{Book-crossing} \cite{ziegler2005improving} contains users' implicit and explicit ratings of books. Each data sample contains a user and a book with their corresponding features. The reprocessed words in the book titles are also regarded as features of the book.
\textbf{MovieLens 25M} \cite{harper2015movielens} contains movie attribute information and the tags (e.g., gunfight, dragon) that users gave to the movies. We regard the attributes genres, years, and the 10 most relevant tags for movie as features.

We transfer the explicit ratings to implicit feedback. We regard the ratings greater than 3 as positive ratings for MovieLens 1M and MovieLens 25M, and regard all rated explicit ratings as positive ratings for Book-crossing due to its sparsity. Then, we randomly select the same number of negative samples equal to the number of positive samples for each user. 

\paragraph{Baselines}
We compare our model with eleven state-of-the-art models that consider feature interactions, including:

\textbf{FM} \cite{koren2008factorization} models every pairwise feature interactions by dot product and sum up all the modeling results as the prediction. 
\textbf{AFM} \cite{xiao2017attentional} calculates an attention value for each feature interaction on top of FM.
\textbf{NFM} \cite{he2017neural} models each pairwise feature interaction and aggregates all the modeling results by an MLP.
\textbf{Fi-GNN} \cite{li2019fi} represents each data sample as a feature graph and models all pairwise interactions in a self-attention based GNN.
\textbf{$L_0$-SIGN} \cite{su2020detecting} represents each data sample as a feature graph. It detects beneficial pairwise feature interactions and uses a GNN to model the detected ones for recommendation.
\textbf{AutoInt} \cite{song2019autoint} explicitly models all pairwise feature interactions using a multi-head self-attentive neural network and aggregate all the modeling results as the prediction.
\textbf{DeepFM} \cite{guo2017deepfm} combines the results of an MLP that implicitly models high-order feature interactions and a FM that models pairwise feature interactions.
\textbf{xDeepFM} \cite{lian2018xdeepfm} is an extension of DeepFM that models high-order feature interactions in explicit way. 
\textbf{DCNv2} \cite{wang2021dcn} leverages deep and cross network to implicitly and explicitly model feature interactions.
\textbf{AutoFIS} \cite{liu2020autofis} detects limited high-order beneficial feature interactions by setting a gate for each feature interaction. 
\textbf{AFN} \cite{cheng2020adaptive} leverages logarithmic neural network to learn high-order feature interactions adaptively from data.  
For all baselines, we use the same MLP settings in all baselines (if use) as our interaction modeling function in HIRS for a fair comparison. 
In terms of he highest order for baselines that explicit model high-order feature interactions, we follow the order settings in their original papers.

\section{Running Time Analysis}
\label{appx:time_complexity}
We have theoretically analyzed the efficiency advantage of HIRS over other baselines that explicitly consider high-order feature interactions. In this section, we experimentally evaluate the running time of HIRS and these baselines. 
Specifically, we run HIRS and baselines on the MovieLens 1M dataset and record the time consuming they used in each epoch while training.  
For the baselines, we record the running time of setting different order limitations. We then compare with HIRS that consider arbitrary order feature interactions. Table \ref{tab:runtime} demonstrates the time consuming results.
Note that since the running environment of our model and baseline models are not exact the same, the time comparison results reported are not completely accurate. For example, some baseline models directly call packages for compiling, which is more time efficient. However, we can still gain useful information from the results.   

From the table, we can see that the running time of our model approximately equals to baseline models modeling low order feature interactions, which is consistent with our theoretical analysis in Section \ref{sec:time_complexity}.
Meanwhile, the running time for baseline models increase with the highest order increasing, while HIRS has constant running time.
The results proves the efficiency of our model in modeling arbitrary order feature interactions, while baseline models have much higher running time when considering up to arbitrary orders (order 14).
In addition, when we detect fewer feature interactions for each data sample, e.g., $k=20$, the running time is further decreased comparing to that of $k=40$.

\begin{table}[t]
\small
\centering
\begin{tabular}{lccccccc}
\hline
\textbf{Method} & \multicolumn{6}{c}{\textbf{Order}}\\
                & 2 & 3 & 4 & 5 & ... & 13 & 14\\
\hline
AutoInt     & 5.48 & 6.26 & 10.96 & 12.59 & ... & 23.34 & 24.27 \\
DCNv2      & 9.40 & 14.09  & 18.79 & 22.71 & ... & 46.19 & 50.11 \\
xDeepFM   & 18.79 & 39.93  & 59.51 & 78.30 & ... & 229.42 & 248.21 \\
\hline
HIRS ($k$=40)       & \multicolumn{7}{c}{15.16} \\
HIRS ($k$=20)       & \multicolumn{7}{c}{11.38} \\
\hline
\end{tabular}
\caption{\small Comparing the time (in seconds) of different models for one epoch on the MovieLens 1M dataset. All the models are run on the same machine equipped with a CPU: Intel(R) Core(TM) i7-9700K @ 3.60GHz, and a GPU: GeForce RTX 2080Ti.}
\label{tab:runtime}
\end{table}

\section{Additional Experimental Results}
\label{appx:additional_exp}

In this section, we list additional results (in Figures \ref{appx:fig1}-\ref{appx:figlast}) that are not shown in our paper due to the space limit.
These results have consistent trends to the results reported in the paper, which provide further evidence on the robustness of our model.

\begin{figure}[ht]
\centering
\begin{tikzpicture}
\begin{axis}[barstyle, xlabel={MovieLens 25M},bar width=1.5mm, symbolic x coords = {HIRS, w/o MI, w/o $L_0$}, legend style={draw=none, at={(0.15,0.99)},anchor=west, nodes={scale=0.9, transform shape}, legend image post style={scale=0.6}}, ymax=0.94,ymin=0.87]
\addplot [color=mycolor4, fill=mycolor4] coordinates { 
(HIRS, 0.9223 )
(w/o MI, 0.9164)
(w/o $L_0$, 0.9182 )
};
\addplot[color=mycolor1, fill=mycolor1] coordinates { 
(HIRS, 0.8846 )
(w/o MI, 0.87997 )  
(w/o $L_0$,  0.88116 ) 
};
\legend{R@10, N@10}
\end{axis}
\end{tikzpicture}
\begin{tikzpicture}
\begin{axis}[barstyle, xlabel={MovieLens 25M},bar width=1.5mm, symbolic x coords = {HIRS, w/o HP, w/o NM, w/o Both}, legend style={draw=none, at={(0.15,0.99)},anchor=west, nodes={scale=0.9, transform shape}, legend image post style={scale=0.6}},ymax=0.94,ymin=0.87]
\addplot [color=mycolor4, fill=mycolor4] coordinates { 
(HIRS,    0.9223)
(w/o HP,  0.9185)
(w/o NM,  0.9174)
(w/o Both, 0.9138)
};
\addplot[color=mycolor1, fill=mycolor1] coordinates { 
(HIRS,    0.8846)
(w/o HP,  0.8818)
(w/o NM,  0.8803)
(w/o Both,0.8778)
};
\legend{R@10, N@10}
\end{axis}
\end{tikzpicture}
\vskip -0.10in
\caption{\textit{Left}: Additional results for Figure \ref{fig:influence_infomax_performance} on the MovieLens 25M dataset. \textit{Right}: Additional results for Figure \ref{fig:influence_component_performance} on the MovieLens 25M dataset.}
\vskip -0.15in
\end{figure}
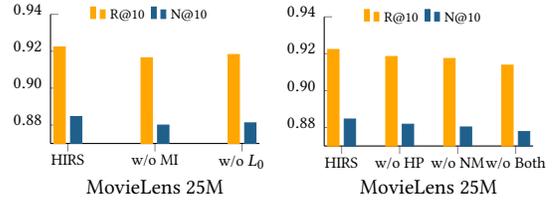

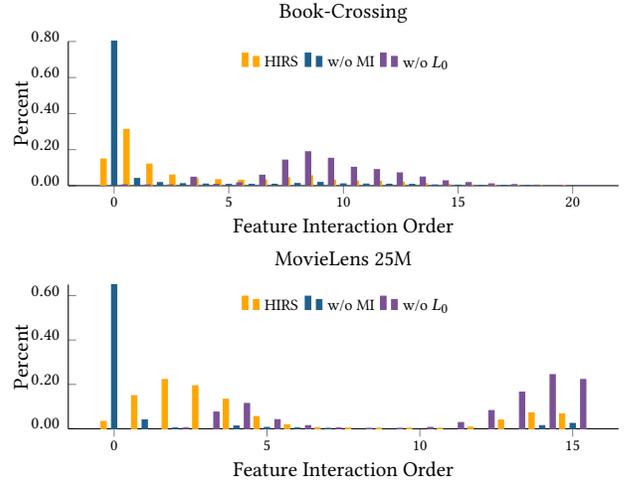
\begin{figure}[ht]
\centering
\begin{tikzpicture}
\begin{axis}[barstyle, width=0.5\textwidth, height=3.5cm, title=Book-Crossing, legend columns=3,ylabel={Percent}, xlabel={Feature Interaction Order}, symbolic x coords = {0,1,2,3,4,5,6,7,8,9,10,11,12,13,14,15,16,17,18,19,20,21,22,23,24,25,26,27,28,29,30,31,32,33,34,35,36,37,38,39}, legend style={draw=none, at={(0.3,0.85)},anchor=west, nodes={scale=1, transform shape}, legend image post style={scale=0.7}},ymin=0, ymax=0.8]
\addplot [y filter/.expression={y==0 ? nan : y},color=mycolor4, fill=mycolor4] coordinates { 
(0, 0.148)
(1, 0.312)
(2, 0.120)
(3, 0.059)
(4, 0.040)
(5, 0.033)
(6, 0.030)
(7, 0.032)
(8, 0.044)
(9, 0.055)
(10, 0.031)
(11, 0.026)
(12, 0.024)
(13, 0.018)
(14, 0.011)
(15, 0.007)
(16, 0.004)
(17, 0.003)
(18, 0.001)
(19, 0.001)
(20, 0.001)
(21, 0)
(22, 0)
(23, 0)
(24, 0)
(25, 0)
(26, 0)
(27, 0)
(28, 0)
(29, 0)
(30, 0)
(31, 0)
(32, 0)
(33, 0)
(34, 0)
(35, 0)
(36, 0)
(37, 0)
(38, 0)
(39, 0)
};
\addplot[y filter/.expression={y==0 ? nan : y},color=mycolor1, fill=mycolor1] coordinates { 
(0, 0.826)
(1, 0.040)
(2, 0.017)
(3, 0.011)
(4, 0.009)
(5, 0.007)
(6, 0.007)
(7, 0.008)
(8, 0.012)
(9, 0.018)
(10, 0.010)
(11, 0.009)
(12, 0.008)
(13, 0.007)
(14, 0.004)
(15, 0.003)
(16, 0.002)
(17, 0.001)
(18, 0.001)
(19, 0)
(20, 0)
(21, 0)
(22, 0)
(23, 0)
(24, 0)
(25, 0)
(26, 0)
(27, 0)
(28, 0)
(29, 0)
(30, 0)
(31, 0)
(32, 0)
(33, 0)
(34, 0)
(35, 0)
(36, 0)
(37, 0)
(38, 0)
(39, 0)
};
\addplot[y filter/.expression={y==0 ? nan : y}, color=mycolor2, fill=mycolor2] coordinates { 
(0, 0.005)
(1, 0.004)
(2, 0.004)
(3, 0.047)
(4, 0.007)
(5, 0.016)
(6, 0.058)
(7, 0.142)
(8, 0.188)
(9, 0.152)
(10, 0.102)
(11, 0.089)
(12, 0.071)
(13, 0.048)
(14, 0.027)
(15, 0.017)
(16, 0.010)
(17, 0.006)
(18, 0.003)
(19, 0.002)
(20, 0.001)
(21, 0)
(22, 0)
(23, 0)
(24, 0)
(25, 0)
(26, 0)
(27, 0)
(28, 0)
(29, 0)
(30, 0)
(31, 0)
(32, 0)
(33, 0)
(34, 0)
(35, 0)
(36, 0)
(37, 0)
(38, 0)
(39, 0)
};
\legend{HIRS, w/o MI, w/o $L_0$}
\end{axis}
\end{tikzpicture}

\begin{tikzpicture}
\begin{axis}[barstyle, width=0.5\textwidth, height=3.5cm, title=MovieLens 25M, legend columns=3,ylabel={Percent}, xlabel={Feature Interaction Order}, symbolic x coords = {0,1,2,3,4,5,6,7,8,9,10,11,12,13,14,15}, legend style={draw=none, at={(0.3,0.85)},anchor=west, nodes={scale=1, transform shape}, legend image post style={scale=0.7}},ymin=0, ymax=0.65]
\addplot [y filter/.expression={y==0 ? nan : y},color=mycolor4, fill=mycolor4] coordinates { 
(0, 0.033)
(1,0.149 )
(2, 0.222 )
(3,0.194  )
(4, 0.133 )
(5, 0.055 )
(6, 0.018 )
(7, 0.005 )
(8, 0.002 )
(9, 0.001 )
(10, 0.001 )
(11, 0.001 )
(12, 0.007 )
(13, 0.039 )
(14, 0.072 )
(15, 0.067 )
};
\addplot[y filter/.expression={y==0 ? nan : y},color=mycolor1, fill=mycolor1] coordinates { 
(0, 0.897 )
(1,0.040 )
(2, 0.003 )
(3, 0.000 )
(4, 0.012 )
(5, 0.006 )
(6, 0.003 )
(7, 0.001 )
(8, 0.000 )
(9, 0.000 )
(10, 0.000 )
(11, 0.000 )
(12, 0.000 )
(13, 0.000 )
(14, 0.014 )
(15, 0.024 )
};
\addplot[y filter/.expression={y==0 ? nan : y}, color=mycolor2, fill=mycolor2] coordinates { 
(0, 0.000 )
(1, 0.000 )
(2, 0.004 )
(3, 0.075 )
(4, 0.114 )
(5, 0.041 )
(6, 0.014 )
(7, 0.003 )
(8, 0.001 )
(9, 0.001 )
(10, 0.006 )
(11, 0.028 )
(12, 0.082 )
(13, 0.165 )
(14, 0.244 )
(15, 0.222 )
};
\legend{HIRS, w/o MI, w/o $L_0$}
\end{axis}
\end{tikzpicture}
\caption{Additional results for Figure \ref{fig:influence_infomax_edgeprediction}. Comparing the distribution of the generated feature interaction orders from HIRS, HIRS without s-Infomax and Infomin, and HIRS without $L_0$ regularization on the Book-Crossing and the MovieLens 25M datasets.}
\label{appx:fig1}
\end{figure}

\begin{figure}[H]
\centering
\begin{tikzpicture}
\begin{axis}[linestyle,width=0.205\textwidth, xlabel={$\lambda_1$},legend columns=2, ymax=0.95, ymin=0.89, symbolic x coords={0.0002, 0.002, 0.02, 0.2, 0.4}, legend style={draw=none, at={(0.9,0.85)},anchor=east, nodes={scale=0.65, transform shape}}, legend image post style={scale=0.6}]
\addplot[mark=*, color=mycolor2] coordinates {
(0.0002,  0.92228 )
(0.002, 0.92505 )
(0.02,  0.9279)
(0.2,  0.920073 )
(0.4, 0.91779 )
};
\addplot[mark=triangle*,color=mycolor3] coordinates {
(0.0002, 0.90248 )
(0.002, 0.9061 )
(0.02, 0.9081)
(0.2, 0.897795 )
(0.4, 0.893806 )
};
\legend{R@10, N@10}
\end{axis}
\end{tikzpicture}
\begin{tikzpicture}
\begin{axis}[linestyle, width=0.205\textwidth, title=Book-Crossing, yticklabels={,,}, xlabel={$\lambda_2$},legend columns=2, ymax=0.95, ymin=0.89, symbolic x coords={0.001, 0.01, 0.1, 1, 10}, legend style={draw=none, at={(0.9,0.85)},anchor=east, nodes={scale=0.65, transform shape}}, legend image post style={scale=0.6}]
\addplot[mark=*, color=mycolor2] coordinates {
(0.001, 0.92426 )
(0.01,  0.92481 )
(0.1,   0.92576 )
(1,     0.9279 )
(10,    0.92761 )
};
\addplot[mark=triangle*,color=mycolor3] coordinates {
(0.001, 0.9031721 )
(0.01,  0.90448 )
(0.1,   0.906231 )
(1,     0.90811)
(10,    0.90774 )
};
\legend{R@10, N@10}
\end{axis}
\end{tikzpicture}
\begin{tikzpicture}
\begin{axis}[linestyle, width=0.205\textwidth, yticklabels={,,}, xlabel={$\lambda_3$},legend columns=2, ymax=0.95, ymin=0.89, symbolic x coords={0.001, 0.01, 0.1, 1, 10}, legend style={draw=none, at={(0.9,0.85)},anchor=east, nodes={scale=0.65, transform shape}}, legend image post style={scale=0.6}]
\addplot[mark=*, color=mycolor2] coordinates {
(0.001, 0.92681)
(0.01,  0.928315)
(0.1, 0.9279 )
(1,  0.9277 )
(10, 0.92523)
};
\addplot[mark=triangle*,color=mycolor3] coordinates {
(0.001, 0.90736)
(0.01,  0.90881)
(0.1, 0.90811)
(1, 0.90887 )
(10, 0.905033)
};
\legend{R@10, N@10}
\end{axis}
\end{tikzpicture}
\begin{tikzpicture}
\begin{axis}[linestyle,width=0.205\textwidth, xlabel={$\lambda_1$},legend columns=2, symbolic x coords={0.0002, 0.002, 0.02, 0.2, 0.4},ymax=0.955, ymin=0.86, legend style={draw=none, at={(0.9,0.85)},anchor=east, nodes={scale=0.65, transform shape}}, legend image post style={scale=0.6}]
\addplot[mark=*, color=mycolor2] coordinates {
(0.0002,  0.9207)
(0.002, 0.9214)
(0.02,  0.9223)
(0.2,  0.9219)
(0.4, 0.9210)
};
\addplot[mark=triangle*,color=mycolor3] coordinates {
(0.0002, 0.8813 )
(0.002, 0.8827)
(0.02, 0.8846)
(0.2, 0.8816)
(0.4, 0.8814)
};
\legend{R@10, N@10}
\end{axis}
\end{tikzpicture}
\begin{tikzpicture}
\begin{axis}[linestyle, width=0.205\textwidth, title=MovieLens 25M,  yticklabels={,,}, xlabel={$\lambda_2$},legend columns=2, ymax=0.955, ymin=0.86, symbolic x coords={0.001, 0.01, 0.1, 1, 10}, legend style={draw=none, at={(0.9,0.85)},anchor=east, nodes={scale=0.65, transform shape}}, legend image post style={scale=0.6}]
\addplot[mark=*, color=mycolor2] coordinates {
(0.001, 0.9188)
(0.01,  0.9201)
(0.1,   0.9206)
(1,     0.9223)
(10,    0.9156)
};
\addplot[mark=triangle*,color=mycolor3] coordinates {
(0.001, 0.8743)
(0.01,  0.8792)
(0.1,   0.8819)
(1,     0.8846)
(10,    0.8837)
};
\legend{R@10, N@10}
\end{axis}
\end{tikzpicture}
\begin{tikzpicture}
\begin{axis}[linestyle, width=0.205\textwidth, yticklabels={,,}, xlabel={$\lambda_3$},legend columns=2, ymax=0.955, ymin=0.86, symbolic x coords={0.001, 0.01, 0.1, 1, 10}, legend style={draw=none, at={(0.9,0.85)},anchor=east, nodes={scale=0.65, transform shape}}, legend image post style={scale=0.6}]
\addplot[mark=*, color=mycolor2] coordinates {
(0.001, 0.9207)
(0.01,  0.9213)
(0.1, 0.9223)
(1,  0.9210)
(10, 0.9196)
};
\addplot[mark=triangle*,color=mycolor3] coordinates {
(0.001, 0.8801)
(0.01,  0.8832)
(0.1, 0.8846)
(1, 0. 8827)
(10,0.8792)
};
\legend{R@10, N@10}
\end{axis}
\end{tikzpicture}
\caption{Additional results for Figure \ref{fig:para_study_weights}. Comparing the performance of different weight values of $L_0$ activation regularization, s-Infomax, and Infomin on the Book-Crossing and the MovieLens 25M dataset.}
\label{appx:figlast}
\end{figure}
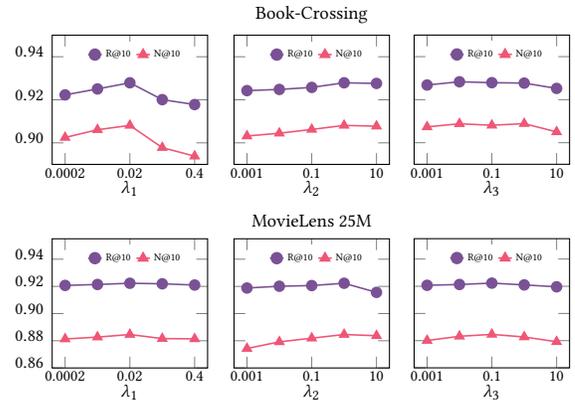

\end{document}